\documentclass[letterpaper, 11pt]{article}

\usepackage[english]{babel}
\usepackage[utf8x]{inputenc}
\usepackage[T1]{fontenc}
\usepackage[margin=1in]{geometry}
\usepackage{amsmath,amsthm,amssymb}
\usepackage{url} 
\usepackage{graphicx}
\usepackage{empheq}
\usepackage{amsfonts}
\usepackage{amsthm}
\usepackage{bbm}
\usepackage[linesnumbered,ruled]{algorithm2e}
\usepackage{caption}
\usepackage{subcaption}
\newtheorem{definition}{Definition}
\newtheorem{theorem}{Theorem}
\newtheorem{lemma}{Lemma}
\newtheorem{corollary}{Corollary}

\newcommand{\argmin}{\mathop{\mathrm{arg\,min}}}
\newcommand{\ignore}[1]{}
\def\Beta{{\frak B}}

\usepackage{authblk}
\allowdisplaybreaks

\renewcommand{\hat}{\widehat}

\usepackage{xspace}

\newcommand{\ie}{{i.e.},\xspace}

\newcommand{\E}{\mathbb{E}}
\renewcommand{\S}{\mathbb{S}}
\newcommand{\proj}{\mathsf{proj}_K}
\usepackage{tikz}

\title{Learning Neural Networks with Two Nonlinear Layers in Polynomial Time}
\author[]{Surbhi Goel\footnote{Supported by University of Texas at Austin Graduate School Summer 2017 Fellowship.} }
\author[]{Adam Klivans\footnote{Supported by NSF Algorithmic Foundations Award AF-1717896.}}
\affil[]{Department of Computer Science, University of Texas at Austin}
\affil[]{\texttt{\{surbhi,klivans\}@cs.utexas.edu}}
\date{}

\begin{document}
\maketitle
\begin{abstract}
  We give a polynomial-time algorithm for learning neural networks
  with one layer of sigmoids feeding into any Lipschitz, monotone
  activation function (e.g., sigmoid or ReLU).  We make no assumptions
  on the structure of the network, and the algorithm succeeds with
  respect to {\em any} distribution on the unit ball in $n$ dimensions
  (hidden weight vectors also have unit norm).  This is the first
  assumption-free, provably efficient algorithm for learning neural
  networks with two nonlinear layers.

  Our algorithm-- {\em Alphatron}-- is a simple, iterative update rule
  that combines isotonic regression with kernel methods.  It outputs a
  hypothesis that yields efficient oracle access to interpretable
  features. It also suggests a new approach to Boolean learning
  problems via real-valued conditional-mean functions, sidestepping
  traditional hardness results from computational learning theory.

  Along these lines, we subsume and improve many longstanding results
  for PAC learning Boolean functions to the more general, real-valued setting of {\em
    probabilistic concepts}, a model that (unlike PAC learning)
  requires non-i.i.d. noise-tolerance.

\end{abstract}

\newpage

\section{Introduction}
Giving provably efficient algorithms for learning neural networks is a fundamental challenge in the theory of machine learning.  Most work in computational learning theory has led to negative results showing that-- from a worst-case perspective-- even learning the simplest architectures seems computationally intractable \cite{LivniSS14,sxvw}. For example, there are known hardness results for agnostically learning a single ReLU (learning a ReLU in the non-realizable setting) \cite{goel2016reliably}.

As such, much work has focused on finding algorithms that succeed after making various restrictive assumptions on both the network's architecture and the underlying marginal distribution.  Recent work gives evidence that for gradient-based algorithms these types of assumptions are actually necessary \cite{shamir2016distribution}.  In this paper, we focus on understanding the frontier of efficient neural network learning: what is the most expressive class of neural networks that can be learned, provably, in polynomial-time without taking any additional assumptions?


 \subsection{Our Results}

We give a simple, iterative algorithm that efficiently learns neural networks with one layer of sigmoids feeding into any smooth, monotone activation function (for example, Sigmoid or ReLU).  Both the first hidden layer of sigmoids and the output activation function have corresponding hidden weight vectors. The algorithm succeeds with respect to any distribution on the unit ball in $n$ dimensions.   The network can have an arbitrary feedforward structure, and we assume nothing about these weight vectors other than that they each have $2$-norm at most one in the first layer (the weight vector in the second layer may have polynomially large norm).  These networks, even over the unit ball, have polynomially large VC dimension (if the first layer has $m$ hidden units, the VC dimension will be $\Omega(m)$ \cite{LBW94}).  

This is the first provably efficient, assumption-free result for learning neural networks with more than one nonlinear layer; prior work due to Goel et al. \cite{goel2016reliably} can learn a sum of one hidden layer of sigmoids.  While our result ``only'' handles one additional nonlinear output layer, we stress that 1) the recent (large) literature for learning even one nonlinear layer often requires many assumptions (e.g., Gaussian marginals) and 2) this additional layer allows us to give broad generalizations of many well-known results in computational learning theory.


Our algorithm, which we call {\em Alphatron}, combines the expressive power of kernel methods with an additive update rule inspired by work from isotonic regression.  Alphatron also outputs a hypothesis that gives efficient oracle access to interpretable features.  That is, if the output activation function is $u$, Alphatron constructs a hypothesis of the form $u(f(\textbf{x}))$ where $f$ is an implicit encoding of products of features from the instance space, and $f$ yields an efficient algorithm for random access to the coefficients of these products.  

More specifically, we obtain the following new supervised learning results:

\begin{itemize}

\item Let $c(\textbf{x}_{1},\ldots,\textbf{x}_{n})$ be any feedforward neural network with one hidden layer of sigmoids of size $k$ feeding into any activation function $u$ that is monotone and $L$-Lipschitz. Given independent draws $(\textbf{x},y)$ from $\S^{n-1} \times [0,1]$ with $\E[y|\textbf{x}] = c(\textbf{x})$, we obtain an efficiently computable hypothesis $u(f(\textbf{x}))$ such that $\E[(c(\textbf{x}) - u(f(\textbf{x})))^2] \leq \epsilon$ with running time and sample complexity $\mathsf{poly}(n,k,1/\epsilon,L)$ (the algorithm succeeds with high probability).  Note that the related (but incomparable) problem of distribution-free PAC learning intersections of halfspaces is cryptographically hard \cite{KS09}. 


\item With an appropriate choice of kernel function, we show that Alphatron can learn more general, real-valued versions of well-studied Boolean concept classes in the {\em probabilistic concept model} due to Kearns and Schapire.  We subsume and improve known algorithms for uniform distribution learning of DNF formulas (queries), majorities of halfspaces, majorities of $\mathsf{AC}^0$ circuits, and submodular functions, among others.  We achieve the first non-i.i.d. noise-tolerant algorithms\footnote{Previously these classes were known to be learnable in the presence of classification noise, where each is label is flipped independently with some fixed probability.} for learning these classes\footnote{Non-iid/agnostic noise tolerance was known for majorities of halfspaces only for $\epsilon < 1/k^2$, where $k$ is the number of halfspace \cite{KKMS}.}. Our technical contributions include

\begin{itemize}

\item Extending the KM algorithm for finding large Fourier coefficients \cite{kushilevitz1993learning} to the setting of probabilistic concepts.  For the uniform distribution on the hypercube, we can combine the KM algorithm's sparse approximations with a projection operator to learn smooth, monotone combinations of $L_1$-bounded functions (it is easy to see that DNF formulas fall into this class).  This improves the approach of Gopalan, Kalai, and Klivans \cite{gopalan2008agnostically} for agnostically learning decision trees.

\item Generalizing the ``low-degree'' algorithm due to Linial, Mansour, and Nisan \cite{LMN93} to show that for any circuit class that can be approximated by low-degree Fourier polynomials, we can learn monotone combinations of these circuits ``for free'' in the probabilistic concept model. 

\item Using low-weight (as opposed to just low-degree) polynomial approximators for intersections of halfspaces with a (constant) margin to obtain the first {\em polynomial-time} algorithms for learning smooth, monotone combinations (intersection is a special case). The previous best result was a quasipolynomial-time algorithm for PAC learning the special case of ANDs of halfspaces with a (constant) margin \cite{klivans2008learning}.


\end{itemize}




\end{itemize}

We also give the first provably efficient algorithms for nontrivial schemes in multiple instance learning (MIL).  Fix an MIL scheme where a learner is given a set of instances $\textbf{x}_{1},\ldots,\textbf{x}_{t}$, and the learner is told only some function of their labels, namely $u(c(\textbf{x}_{1}),\ldots,c(\textbf{x}_{t}))$ for some unknown concept $c$ and monotone combining function $u$.  We give the first provably efficient algorithms for correctly labeling future bags even if the instances within each bag are not identically distributed.  Our algorithms hold if the underlying concept $c$ is sigmoidal or a halfspace with a margin.  If the combining function averages label values (a common case), we obtain bounds that are {\em independent} of the bag size.



We learn specifically with respect to square loss, though this will imply polynomial-time learnability for most commonly studied loss functions.  When the label $Y$ is a deterministic Boolean function of $X$, it is easy to see that small square loss will imply small $0/1$ loss.

\subsection{Our Approach}

The high-level approach is to use algorithms for isotonic regression to learn monotone combinations of functions approximated by elements of a suitable RKHS.  Our starting point is the Isotron algorithm, due to Kalai and Sastry \cite{KalaiS09}, and a refinement due to Kakade, Kalai, Kanade and Shamir \cite{KKKS11} called the GLMtron.  These algorithms efficiently learn any generalized linear model (GLM): distributions on instance-label pairs $(\textbf{x},y)$ where the conditional mean of $y$ given $\textbf{x}$ is equal to $u(\textbf{w} \cdot \textbf{x})$ for some (known) smooth, non-decreasing function $u$ and unknown weight vector $\textbf{w}$.  Their algorithms are simple and use an iterative update rule to minimize square-loss, a non-convex optimization problem in this setting.  Both of their papers remark that their algorithms can be kernelized, but no concrete applications are given.

Around the same time, Shalev-Shwartz, Shamir, and Sridharan \cite{SSSS} used kernel methods and general solvers for convex programs to give algorithms for learning a halfspace under a distributional assumption corresponding to a margin in the non-realizable setting (agnostic learning).  Their kernel was composed by Zhang et al.~\cite{Zhang} to obtain results for learning sparse neural networks with certain smooth activations, and Goel et al.~\cite{goel2016reliably} used a similar approach in conjunction with general tools from approximation theory to obtain learning results for a large class of nonlinear activations including ReLU and Sigmoid.

Combining the above approaches, though not technically deep, is subtle and depends heavily on the choice of model.  For example, prior work on kernel methods for learning neural networks has focused almost exclusively on learning in the agnostic model.  This model is {\em too} challenging, in the sense that the associated optimization problems to be solved seem computationally intractable (even for a single ReLU).  The probabilistic concept model, on the other hand, is a more structured noise model and allows for an iterative approach to minimize the empirical loss.

Our algorithm-- {\em Alphatron}-- inherits the best properties of both kernel methods and gradient-based methods: it is a simple, iterative update rule that does not require regularization\footnote{We emphasize this to distinguish our algorithm from the usual kernel methods (e.g., kernel ridge regression and SVMs) where regularization and the representer theorem are key steps.}, and it learns broad classes of networks whose first layer can be approximated via an appropriate feature expansion into an RKHS.  


One technical challenge is handling the approximation error induced from embedding into an RKHS.  In some sense, we must learn a {\em noisy} GLM.  For this, we use a learning rate and a slack variable to account for noise and follow the outline of the analysis of GLMtron (or Isotron).  The resulting algorithm is similar to performing gradient descent on the support vectors of a target element in an RKHS.  Our convergence bounds depend on the resulting choice of kernel, learning rate, and quality of RKHS embedding.  We can then leverage several results from approximation theory and obtain general theorems for various notions of RKHS approximation.

\subsection{Related Work}

The literature on provably efficient algorithms for learning neural networks is extensive.  In this work we focus on common nonlinear activation functions: sigmoid, ReLU, or threshold.  For linear activations, neural networks compute an overall function that is linear and can be learned efficiently using any polynomial-time algorithm for solving linear regression.  Livni et al. \cite{LivniSS14} observed that neural networks of constant depth with constant degree polynomial activations are equivalent to linear functions in a higher dimensional space (polynomials of degree $d$ are equivalent to linear functions over $n^{d}$ monomials).  It is known, however, that any polynomial that computes or even $\epsilon$-approximates a single ReLU requires degree $\Omega(1/\epsilon)$ \cite{goel2016reliably}.  Thus, linear methods alone do not suffice for obtaining our results.

The vast majority of work on learning neural networks takes strong assumptions on either the underlying marginal distribution (e.g., Gaussian), the structure of the network, or both.  Works that fall into these categories include \cite{KliODoSer04,KlivansM13,janzamin2015beating,sedghi2014provable,ZhangPS17,Zhang,Zhao,GoelKlivans}.  In terms of assumption-free learning results, Goel et al. \cite{goel2016reliably} used kernel methods to give an efficient, agnostic learning algorithm for sums of sigmoids (i.e., one hidden layer of sigmoids) with respect to any distribution on the unit ball. Daniely \cite{Daniely17} used kernel methods in combination with gradient descent to learn neural networks, but the networks he considers have restricted VC dimension.  All of the problems we consider in this paper are non-convex optimization problems, as it is known that a single sigmoid with respect to square-loss has exponentially many bad local minima \cite{AHW}. \\

{\noindent}{\bf A Remark on Bounding the 2-Norm}. 
As mentioned earlier, the networks we learn, even over the unit ball, have polynomially large VC dimension (if the first layer has $m$ hidden units, the VC dimension will be $\Omega(m) \cite{LBW94}$).  It is easy to see that if we allow the $2$-norm of weight vectors in the first layer to be polynomially large (in the dimension), we arrive at a learning problem statistically close to PAC learning intersections of halfspaces, for which there are known cryptographic hardness results \cite{KS09}.  Further, in the agnostic model, learning even a single ReLU with a bounded norm weight vector (and any distribution on the unit sphere) is as hard as learning sparse parity with noise \cite{goel2016reliably}.  As such, for distribution-free learnability, it seems necessary to have some bound on the norm {\em and} some structure in the noise model.  Bounding the norm of the weight vectors also aligns nicely with practical tools for learning neural networks.  Most gradient-based training algorithms for learning deep nets initialize hidden weight vectors to have unit norm and use techniques such as batch normalization or regularization to prevent the norm of the weight vectors from becoming large.





\subsection{Notation}
Vectors are denoted by bold-face and $|| \cdot ||$ denotes the standard 2-norm of the vector. We denote the space of inputs by $\mathcal{X}$ and the space of outputs by $\mathcal{Y}$. In our paper, $\mathcal{X}$ is usually the unit sphere/ball and $\mathcal{Y}$ is $[0,1]$ or $\{0,1\}$. Standard scalar (dot) products are denoted by $\textbf{a} \cdot \textbf{b}$ for vectors $\textbf{a}, \textbf{b} \in\mathbb{R}^n$, while inner products in a Reproducing Kernel Hilbert Space (RKHS) are denoted by $\langle \textbf{a}, \textbf{b} \rangle$ for elements $\textbf{a},\textbf{b}$ in the RKHS. We denote the standard composition of functions $f_1$ and $f_2$ by $f_1 \circ f_2$.\\

\noindent\textbf{Note.} Due to space limitations, we defer most proofs to the appendix.


\section{The Alphatron Algorithm}
Here we present our main algorithm Alphatron (Algorithm \ref{alphatron}) and a proof of its correctness.  In the next section we will use this algorithm to obtain our most general learning results.

\begin{algorithm}
\caption{Alphatron}\label{alphatron}
    \SetKwInOut{Input}{Input}
    \SetKwInOut{Output}{Output}
    \Input{data $\langle (\textbf{x}_\textbf{i}, y_i \rangle_{i=1}^m \in \mathbb{R}^n \times [0,1]$, non-decreasing\footnotemark $L$-Lipschitz function $u: \mathbb{R} \rightarrow [0,1]$, kernel function $\mathcal{K}$ corresponding to feature map $\psi$, learning rate $\lambda > 0$, number of iterations $T$, held-out data of size $N$ $\langle \textbf{a}_\textbf{j}, b_j \rangle_{j=1}^N \in \mathbb{R}^n \times [0,1]$}
    $\alpha^1 := 0 \in \mathbb{R}^m$\\
    \For{$t = 1, \ldots, T$}{
    $h^t(\textbf{x}) := u(\sum_{i=1}^m \alpha_i^t \mathcal{K}(\textbf{x}, \textbf{x}_\textbf{i}))$
    \For{$i = 1,2, \ldots, m$}{
    $\alpha^{t+1}_i := \alpha^t_i + \frac{\lambda}{m}(y_i - h^t(\textbf{x}_\textbf{i}))$
    }
    }
    \Output{$h^r$ where $r = \argmin_{t \in \{1, \ldots, T\}} \sum_{j=1}^N(h^t(\textbf{a}_\textbf{j}) - b_j)^2$}
\end{algorithm}

Define $\textbf{v}^\textbf{t} = \sum_{i=1}^m \alpha_i^t \psi(\textbf{x}_\textbf{i})$ implying $h^t(\textbf{x}) = u(\langle \textbf{v}^\textbf{t}, \psi(\textbf{x})\rangle)$. Let $\varepsilon(h) = \mathbb{E}_{\textbf{x},y}[(h(\textbf{x}) - \mathbb{E}[y|\textbf{x}])^2]$ and $err(h) = \mathbb{E}_{\textbf{x},y}[(h(\textbf{x}) - y)^2]$. It is easy to see that $\varepsilon(h) = err(h) - err(\mathbb{E}[y|\textbf{x}])$. Let $\hat{\varepsilon}, \hat{err}$ be the empirical versions of the same.

The following theorem generalizes Theorem 1 of \cite{KKKS11} to the bounded noise setting in a high dimensional feature space.  We follow the same outline, and their theorem can be recovered by setting $\psi(\textbf{x}) = \textbf{x}$ and $\xi$ as the zero function.
\begin{theorem} \label{thm:alpha}
Let $\mathcal{K}$ be a kernel function corresponding to feature map $\psi$ such that $\forall \textbf{x} \in \mathcal{X}, ||\psi(\textbf{x})|| \leq 1$. Consider samples $(\textbf{x}_\textbf{i},y_i)_{i=1}^m$ drawn iid from distribution $\mathcal{D}$ on $\mathcal{X} \times [0,1]$ such that $E[y|\textbf{x}] = u(\langle \textbf{v}, \psi(\textbf{x}) \rangle + \xi(\textbf{x}))$ where $u: \mathbb{R} \rightarrow [0,1]$ is a known $L$-Lipschitz non-decreasing function, $\xi: \mathbb{R}^n \rightarrow [-M, M]$ for $M > 0$ such that $\E[\xi(\textbf{x})^2] \leq \epsilon$ and $||\textbf{v}|| \leq B$. Then for $\delta \in (0,1)$, with probability $1 - \delta$, Alphatron with $\lambda = 1/L, T = CBL\sqrt{m/\log(1/\delta)}$ and $N = C^\prime m\log(T/\delta))$ for large enough constants $C, C^\prime>0$ outputs a hypothesis $h$ such that,
\[
\varepsilon(h) \leq O\left(L\sqrt{\epsilon} + L M \sqrt[4]{\frac{\log(1/\delta)}{m}} + BL\sqrt{\frac{\log(1/\delta)}{m}}\right).
\]
\end{theorem}
\footnotetext{We present the algorithm and subsequent results for non-decreasing function $u$. Non-increasing functions can also be handled by negating the update term, \ie $\alpha^{t+1}_i := \alpha^t_i - \frac{\lambda}{m}(y_i - h^t(\textbf{x}_\textbf{i}))$.}

Alphatron runs in time $\mathsf{poly}(n, m, \log(1/\delta), t_{\mathcal{K}})$ where $t_{\mathcal{K}}$ is the time required to compute the kernel function $\mathcal{K}$.

\subsection{General Theorems Involving Alphatron}
In this section we use Alphatron to give our most general learnability results for the probablistic concept (p-concept) model.  We then state several applications in the next section.  Here we show that if a function can be approximated by an element of an appropriate RKHS, then it is p-concept learnable. We assume that the kernel function is efficiently computable, that is, computable in polynomial time in the input dimension. Formally, we define approximation as follows:
\begin{definition}[$(\epsilon, B, M)$-approximation]
\label{def_epsM}
Let $f$ be a function mapping domain $\mathcal{X}$ to $\mathbb{R}$ and $\mathcal{D}$ be a distribution over $\mathcal{X}$. Let $\mathcal{K}$ be a kernel function with corresponding RKHS $\mathcal{H}$ and feature vector $\psi$. We say $f$ is $(\epsilon, B, M)$-approximated by $\mathcal{K}$ over $\mathcal{D}$ if there exists some $\textbf{v} \in \mathcal{H}$ with $||\textbf{v}|| \leq B$ such that for all $\textbf{x} \in \mathcal{X}, |f(\textbf{x}) - \langle \textbf{v}, \psi(\textbf{x})\rangle| \leq M$ and $\E[(f(\textbf{x}) - \langle \textbf{v}, \psi(\textbf{x})\rangle)^2] \leq \epsilon^2$.
\end{definition}

Combining Alphatron and the above approximation guarantees, we have the following general learning results:
\begin{theorem}\label{thm:main1}
Consider distribution $\mathcal{D}$ on $\mathcal{X} \times [0,1]$ such that $\E[y|\textbf{x}] = u(f(\textbf{x}))$ where $u$ is a known $L$-Lipschitz non-decreasing function and $f$ is $(\epsilon, B, M)$-approximated over $\mathcal{D}_{\mathcal{X}}$ by some kernel function $K$ such that $\mathcal{K}(\textbf{x}, \textbf{x}^\prime) \leq 1$. Then for $\delta \in (0,1)$, there exists an algorithm that draws $m$ iid samples from $\mathcal{D}$ and outputs a hypothesis $h$ such that with probability $1 - \delta$, $\varepsilon(h) \leq O(L\epsilon)$ for $m = O\left(\left(\frac{LM}{\epsilon}\right)^4 + \left(\frac{BL}{\epsilon}\right)^2\right)\cdot\log(1/\delta)$ in time $\mathsf{poly}(n,B,M,L,1/\epsilon,\log(1/\delta))$ where $n$ is the dimension of $\mathcal{X}$.
\end{theorem}
\begin{proof} 
Let $\mathcal{H}$ be the RKHS corresponding to $\mathcal{K}$ and $\psi$ be the feature vector. Since $f$ is $(\epsilon, B, M)$-approximated by kernel function $\mathcal{K}$ over $\mathcal{D}_{\mathcal{X}}$, we have $\forall\ \textbf{x}, f(\textbf{x}) = \langle \textbf{v}, \psi(\textbf{x}) \rangle  + \xi(\textbf{x})$ for some function $\xi: \mathcal{X} \rightarrow [-M, M]$ with $\E[\xi(\textbf{x})^2] \leq \epsilon^2$. Thus $E[y|\textbf{x}] = u(f(\textbf{x})) =  u\left(\langle \textbf{v}, \psi(\textbf{x}) \rangle + \xi(\textbf{x}) \right)$. Applying Theorem \ref{thm:alpha}, we have that Alphatron outputs a hypothesis $h$ such that
\[
\varepsilon(h) \leq CL \left(\epsilon + M \sqrt[4]{\frac{\log(1/\delta)}{m}} + B \sqrt{\frac{\log(1/\delta)}{m}}\right)
\]
for some constants $C > 0$. Also Alphatron requires at most $O(BL\sqrt{m/\log(1/\delta)})$ iterations. Setting $m$ as in theorem statement gives us the required result.
\end{proof}

For the simpler case when $f$ is {\em uniformly} approximated by elements in the RKHS we have,

\begin{definition}[$(\epsilon, B)$-uniform approximation]
\label{def_epsM}
Let $f$ be a function mapping domain $\mathcal{X}$ to $\mathbb{R}$ and $\mathcal{D}$ be a distribution over $\mathcal{X}$. Let $\mathcal{K}$ be a kernel function with corresponding RKHS $\mathcal{H}$ and feature vector $\psi$. We say $f$ is $(\epsilon, B)$-uniformly approximated by $\mathcal{K}$ over $\mathcal{D}$ if there exists some $\textbf{v} \in \mathcal{H}$ with $||\textbf{v}|| \leq B$ such that for all $\textbf{x} \in \mathcal{X}, |f(\textbf{x}) - \langle \textbf{v}, \psi(\textbf{x})\rangle| \leq \epsilon$.
\end{definition}

\begin{theorem}\label{thm:main}
 Consider distribution $\mathcal{D}$ on $\mathcal{X} \times [0,1]$ such that $\E[y|\textbf{x}] = u(f(\textbf{x}))$ where $u$ is a known $L$-Lipschitz non-decreasing function and $f$ is $(\epsilon, B)$-approximated by some kernel function $\mathcal{K}$ such that $\mathcal{K}(\textbf{x}, \textbf{x}^\prime) \leq 1$. Then for $\delta \in (0,1)$, there exists an algorithm that draws $m$ iid samples from $\mathcal{D}$ and outputs a hypothesis $h$ such that with probability $1 - \delta$, $\varepsilon(h) \leq O(L\epsilon)$ for $m \geq \left(\frac{BL}{\epsilon}\right)^2\cdot\log(1/\delta)$ in time $\mathsf{poly}(n,B,L,1/\epsilon,\log(1/\delta))$ where $n$ is the dimension of $\mathcal{X}$.
\end{theorem}

\begin{proof}
The proof is the same as the proof of Theorem \ref{thm:main1} by re-examining the proof of Theorem 1 and noticing that $\xi(\textbf{x})$ is uniformly bounded by $\epsilon$ in each inequality. 
\end{proof}
\section{Learning Neural Networks}
In this section we give polynomial time learnability results for neural networks with two nonlinear layers in the p-concept model. Following Safran and Shamir \cite{SafranS16}, we define a neural network with one (nonlinear) layer with $k$ units as follows:
\[
\mathcal{N}_1: \textbf{x} \rightarrow \sum_{i=1}^k \textbf{b}_i \sigma(\textbf{a}_\textbf{i} \cdot \textbf{x})
\]
for $\textbf{x} \in \mathbb{R}^n$, $\textbf{a}_\textbf{i} \in \S^{n-1}$ for $i \in \{1, \cdots, k\}$, $\textbf{b} \in \S^{k-1}$. We subsequently define a neural network with two (nonlinear) layers with one unit in layer 2 and $k$ units in hidden layer 1 as follows:
\[
\mathcal{N}_2: \textbf{x} \rightarrow \sigma'\left(\mathcal{N}_1(x)\right) = \sigma'\left(\sum_{i=1}^k \textbf{b}_i \sigma(\textbf{a}_\textbf{i} \cdot \textbf{x})\right)
\]
for $\textbf{x} \in \mathbb{R}^n$, $\textbf{a}_\textbf{i} \in \S^{n-1}$ for $i \in \{1, \cdots, k\}$, $\textbf{b} \in \S^{k-1}$ and $\sigma, \sigma': \mathbb{R} \rightarrow \mathbb{R}$.


The following theorem is our main result for learning classes of neural networks with two nonlinear layers in polynomial time: 
\begin{theorem}\label{thm:sss}
Consider samples $(\textbf{x}_\textbf{i},y_i)_{i=1}^m$ drawn iid from distribution $\mathcal{D}$ on $\S^{n-1} \times [0,1]$ such that $E[y|\textbf{x}] = \mathcal{N}_2(\textbf{x})$ with $\sigma': \mathbb{R} \rightarrow [0,1]$ is a known $L$-Lipschitz non-decreasing function and $\sigma =\sigma_{sig}$ is the sigmoid function. There exists an algorithm that outputs a hypothesis $h$ such that, with probability $1-\delta$,
\[
\E_{\textbf{x}, y \sim \mathcal{D}}\left[ (h(\textbf{x}) - \mathcal{N}_2(\textbf{x}))^2\right] \leq \epsilon
\]
for $m = \left(\frac{kL}{\epsilon}\right)^{O(1)}\cdot\log(1/\delta)$.  The algorithm runs in time polynomial in $m$ and $n$.
\end{theorem}

We also obtain results for networks of ReLUs, but the dependence on the number of hidden units, $\epsilon$, and $L$ are exponential (the algorithm still runs in polynomial-time in the dimension):

\begin{theorem}
Consider samples $(\textbf{x}_\textbf{i},y_i)_{i=1}^m$ drawn iid from distribution $\mathcal{D}$ on $\S^{n-1} \times [0,1]$ such that $E[y|\textbf{x}] = \mathcal{N}_2(\textbf{x})$ with $\sigma': \mathbb{R} \rightarrow [0,1]$ is a known $L$-Lipschitz non-decreasing function and $\sigma = \sigma_{relu}$ is the ReLU function. There exists an algorithm that outputs a hypothesis $h$ such that with probability $1 - \delta$, 
\[
\E_{\textbf{x}, y \sim \mathcal{D}}\left[ (h(\textbf{x}) - \mathcal{N}_2(\textbf{x}))^2\right] \leq \epsilon
\]
for $m = 2^{O(k L/\epsilon)} \cdot \log(1/\delta)$. The algorithm runs in time polynomial in $m$ and $n$.
\end{theorem}

Although our algorithm does not recover the parameters of the network, it still outputs a hypothesis with interpretable features. More specifically, our learning algorithm outputs the hidden layer as a multivariate polynomial. Given inputs $\textbf{x}_\textbf{1}, \cdots, \textbf{x}_\textbf{m}$, the hypothesis output by our algorithm Alphatron is of the form $h(\textbf{x}) = u(\sum_{i=1}^m \alpha_i^* \mathcal{MK}_d(\textbf{x}, \textbf{x}_\textbf{i})) = u(\langle \textbf{v}, \psi_d(\textbf{x})\rangle)$ where $\textbf{v} = \sum_{i=1}^m \alpha_i^*\psi_d(\textbf{x}_\textbf{i})$ and $d$ is dependent on required approximation. As seen in the preliminaries, $\langle \textbf{v}, \psi_d(\textbf{x})\rangle$ can be expressed as a polynomial and the coefficients can be computed as follows,
\[ \beta(i_1, \ldots, i_n) = \sum_{i=1}^m \alpha_i^* \left( \sum_{\substack{k_1, \ldots, k_j \in [n]^j
\\ j \in \{0, \ldots, d\} \\ M(k_1, \ldots, k_j ) = (i_1, \ldots, i_n)}}
(x_i)_{k_1} \cdots (x_i)_{k_j}\right) = \sum_{i=1}^m \alpha_i^* C\left(i_1, \ldots,
i_n\right) (x_i)_{1}^{i_1} \cdots (x_i)_{n}^{i_n}. \]
Here, we follow the notation from \cite{goel2016reliably}; $M$ maps ordered tuple $({k_1}, \ldots, {k_j}) \in
[n]^j$ for $j \in [d]$ to tuple $(i_1, \ldots, i_n) \in \{0, \ldots, d\}^n$
such that $x_{k_1}\cdots x_{k_j} = x_1^{i_1} \cdots x_n^{i_n}$ and $C$ maps ordered tuple $(i_1,\ldots, i_n) \in \{0, \ldots, d\}^n$ to the number of distinct orderings of the $i_j$'s for $j \in \{0, \ldots, n\}$. The function $C$ can be computed from the multinomial theorem (cf. \cite{wiki}).  Thus, the coefficients of the polynomial can be efficiently indexed.  Informally, each coefficient can be interpreted as the correlation between the target function and the product of features appearing in the coefficient's monomial.


\section{Generalizing PAC Learning to Probabilistic Concepts}
In this section we show how known algorithms for PAC learning boolean concepts can be generalized to the probabilistic concept model. We use Alphatron to learn real-valued versions of these well-studied concepts.

\noindent\textbf{Notation.} We follow the notation of \cite{gopalan2008agnostically}. For any function $P: \{-1, 1\}^n \rightarrow \mathbb{R}$, we denote the Fourier coefficients by $\hat{P}(S)$ for all $S \subseteq[n]$. The support of $P$, \ie the number of non-zero Fourier coefficients, is denoted by $\mathsf{supp}(P)$. The norms of the coefficient vectors are defined as $L_p(P) = \left(\sum_S |\hat{P}(S)|^p\right)^{1/p}$ for $p \geq 1$ and $L_\infty(P) = \max_S |\hat{P}(S)|$. Similarly, the norm of the function $P$ are defined as $||P||_p = \E_{x \in \{-1,1\}^n}\left[\sum_S |P(x)|^p\right]^{1/p}$ for $p \geq 1$. Also, the inner product $P \cdot Q = \E_{x \in \{-1,1\}^n}[P(x)Q(x)]$.
\subsection{Generalized DNF Learning with Queries}
Here we give an algorithm, KMtron, which combines isotonic regression with the KM algorithm \cite{kushilevitz1993learning} for finding large Fourier coefficients of a function (given query access to the function).  The KM algorithm takes the place of the ``kernel trick'' used by Alphatron to provide an estimate for the update step in isotonic regression.  Viewed this way, the KM algorithm can be re-interpreted as a query-algorithm for giving estimates of the gradient of square-loss with respect to the uniform distribution on Boolean inputs.

The main application of KMtron is a generalization of celebrated results for PAC learning DNF formulas \cite{Jackson} to the setting of probabilistic concepts.  That is, we can efficiently learn any conditional mean that is a smooth, monotone combination of $L_1$-bounded functions. \\

\noindent\textbf{KM Algorithm.}
The KM algorithm learns sparse approximations to boolean functions given query access to the underlying function. The following lemmas about the KM algorithm are important to our analysis.
\begin{lemma}[\cite{kushilevitz1993learning}]\label{lem:km}
Given an oracle for $P: \{-1,1\}^n \rightarrow \mathbb{R}$, $\mathsf{KM}(P, \theta)$ returns $Q: \{-1,1\}^n \rightarrow \mathbb{R}$ with $|\mathsf{supp}(Q)| \leq O(L_2(P)^2\theta^{-2})$ and $L_\infty(P-Q) \leq \theta$. The running time is $\mathsf{poly}(n, \theta^{-1}, L_2(P))$. 
\end{lemma}
\begin{lemma}[\cite{kushilevitz1993learning}]\label{lem:kmsparse}
If $P$ has $L_1(P) \leq k$, then $\mathsf{KM}\left(P, \frac{\epsilon^2}{2k}\right)$ returns $Q$ s.t. $||P-Q||_2 \leq \epsilon$.
\end{lemma}

\noindent \textbf{Projection Operator.}
The projection operator $\proj(P)$ for $P: \{-1,1\}^n \rightarrow \mathbb{R}$ maps $P$ to the closest $Q$ in convex set $K = \{Q: \{-1,1\}^n \rightarrow \mathbb{R}\ |\  L_1(Q) \leq k\}$, \ie $\proj(P) = \argmin_{Q \in K}||Q - P||_2$. \cite{gopalan2008agnostically} show that $\proj$ is simple and easy to compute for sparse polynomials. We use the following lemmas by \cite{gopalan2008agnostically} about the projection operator in our analysis.

\begin{lemma}[\cite{gopalan2008agnostically}]\label{lem:proj1}
Let $P, P^\prime$ be such that $L_\infty(P- P^\prime) \leq \epsilon$. Then $L_\infty(\proj(P) - \proj(P^\prime)) \leq 2 \epsilon$.
\end{lemma}
\begin{lemma}[\cite{gopalan2008agnostically}]\label{lem:proj2}
Let $P, P^\prime$ be such that $L_\infty(P- P^\prime) \leq \epsilon$. Then $||\proj(P) - \proj(P^\prime)||_2 \leq 2\sqrt{\epsilon k}$.
\end{lemma}

\noindent \textbf{KMtron.}
The algorithm KMtron is as follows:

\begin{algorithm}[H]
\caption{KMtron}\label{kmtron}
    \SetKwInOut{Input}{Input}
    \SetKwInOut{Output}{Output}
    \Input{Function $u: \mathbb{R} \rightarrow [0,1]$ non-decreasing and $L$-Lipschitz, query access to $u \circ P$ for some function $P: \{-1,1\}^n \rightarrow \mathbb{R}$, learning rate $\lambda \in (0,1]$, number of iterations $T$, error parameter $\theta$}
    $P_0 = 0$\\
    \For{$t = 1, \ldots, T$}{
    $P_t^\prime := P_{t-1} + \lambda \mathsf{KM}(u \circ P - u \circ P_{t-1}, \theta)$\\
    $P_t = \mathsf{KM}(\proj(P_t^\prime), \theta)$
    }
    \Output{Return $u \circ P_t$ where $P_t$ is the best over $t = 1, \cdots, T$}
\end{algorithm}

To efficiently run KMtron, we require efficient query access to $u \circ P - u \circ P_{t-1}$. Since $P_{t-1}$ is stored as a sparse polynomial, and we are given query access for $u \circ P$, we can efficiently compute $u(P(x)) - u(P_{t-1}(x))$ for any $x$. We can extend the algorithm to handle distribution queries (p-concept), \ie for any $x$ of our choosing we obtain a sample of $y$ where $\E[y|x] = u(P(x))$. \cite{gopalan2008agnostically} (c.f. Appendix A.1) observed that using distribution queries instead of function queries to the conditional mean is equivalent as long as the number of queries is polynomial.

The following theorem proves the correctness of KMtron.
\begin{theorem}\label{thm:kmtron}
For any non-decreasing $L$-Lipschitz $u:\mathbb{R}\rightarrow [0,1]$ and function $P: \{-1,1\}^n \rightarrow \mathbb{R}$ such that $L_1(P) \leq k$, given query access to $u \circ P$, KMtron run with $\lambda = \frac{\epsilon}{2L}, T = \frac{2k^2L^2}{\epsilon^2}$ and $\theta \leq \frac{C^\prime\epsilon^4}{L^4 k^3}$ for sufficiently small constant $C^\prime >0$ and outputs $P^*$ such that $\E_{x \in \{-1,1\}^n}[(u(P(x)) - u(P^*(x)))^2] \leq \epsilon$. The runtime of KMtron is $\mathsf{poly}(n, k, L, 1/\epsilon)$.
\end{theorem}

\begin{corollary} Let $P_i$ be such that $L_1(P_i) \leq k$ for $i \in [s]$. If we have query access to $y$ for all $x$ such that $\E[y|x] = u\left(\frac{1}{s}\sum_{i=1}^s P_i\right)$ for non-decreasing $L$-Lipschitz $u$, then using the above, we can learn the conditional mean function in time $\mathsf{poly}(n, k, L, 1/\epsilon)$ with respect to the uniform distribution on $\{-1,1\}^n$. 
\end{corollary}

Observe that the complexity bounds are \emph{independent of the number of terms}. This follows from the fact that $L_1\left(\frac{1}{s}\sum_{i=1}^s P_i\right) \leq k$.  This leads to the following new learning result for DNF formulas: fix a DNF $f$ and let $\mathsf{frac}(f(x))$ denote the fraction of terms of $f$ satisfied by $x$.  Fix monotone, $L$-Lipschitz function $u$.  For uniformly chosen input $x$, label $y$ is equal to $1$ with probability $u(\mathsf{frac}(f(x)))$.  Then in time polynomial in $n$, $1/\epsilon$, and $L$, KMtron outputs a hypothesis $h$ such that $\E[(h(x) - u(\mathsf{frac}(f(x))))^2] \leq \epsilon$ (recall $L_1 (\mathsf{AND}) = 1$).  Note that the running time has no dependence on the number of terms.


As an easy corollary, we also obtain a simple (no Boosting required) polynomial time query-algorithm for learning DNFs under the uniform distribution\footnote{Feldman \cite{FeldmanSQ} was the first to obtain a query-algorithm for PAC learning DNF formulas with respect to the uniform distribution that did not require a Boosting algorithm.}:

\begin{corollary} \label{cor:dnf}
Let $f$ be a DNF formula from $\{-1,1\}^n \rightarrow \{0,1\}$ with $s$ terms. Then $f$ is PAC learnable under the uniform distribution using membership queries in time $\mathsf{poly}(n, s, 1/\epsilon)$.
\end{corollary}

\subsection{Extending the ``Low-Degree'' Algorithm}

Here we show that Alphatron can be used to learn any smooth, monotone combination of function classes that are approximated by {\em low-degree} polynomials (our other results require us to take advantage of {\em low-weight} approximations).  

\begin{definition}
For a class of functions ${\cal C}$, let $u({\cal C})$ denote monotone function $u$ applied to a linear combination of (polynomially many) functions from ${\cal C}$.
\end{definition}

For the domain of $\{-1,1\}^n$ and degree parameter $d$, our algorithm will incur a sample complexity and running time factor of $n^{d}$, so the ``kernel trick'' is not necessary (we can work explicitly in the feature space).  The main point is that using isotonic regression (as opposed to the original ``low-degree'' algorithm due to Linial, Mansour and Nisan \cite{LMN93}), we can learn $u({\cal C})$ for any smooth, monotone $u$ and class ${\cal C}$ that has low-degree Fourier approximations (we also obtain non-i.i.d. noise tolerance for these classes due to the definition of the probabilistic concept model).  While isotonic regression has the flavor of a boosting algorithm, we do not need to change the underlying distribution on points or add noise to the labels, as all boosting algorithms do.  

\begin{definition}{$(\epsilon, d)$-Fourier concentration}
A function $f:\{-1,1\}^n \rightarrow \mathbb{R}$ is said to be $(\epsilon, d)$-Fourier concentrated if $\sum_{S:|S| > d} \hat{f}(S)^2 \leq \epsilon^2$ where $\hat{f}(S)$ for all $S \subseteq [n]$ are the discrete Fourier coefficients of $f$.
\end{definition}

\begin{theorem} \label{thm:four}
Consider distribution $\mathcal{D}$ on $\{-1,1\}^n\times[0,1]$ whose marginal is uniform on $\{-1,1\}^n$ and $\E[y|\textbf{x}] = u(f(\textbf{x}))$ for some known non-decreasing $L$-Lipschitz $u: \mathbb{R} \rightarrow [0,1]$ and $f:\{-1,1\}^n \rightarrow [-M,M]$. If $f$ is $(\epsilon, d)$-Fourier concentrated then there exists an algorithm that draws $m$ iid samples from $\mathcal{D}$ and outputs hypothesis $h$ such that with probability $1- \delta$, $\varepsilon(h) \leq O(L \epsilon)$ for $m = \mathsf{poly}(n^d, M, \epsilon, \log(1/\delta))$ in time $\mathsf{poly}(n^d, M, \epsilon, \log(1/\delta))$.
\end{theorem}

The above can be generalized to linear combinations of Fourier concentrated functions using the following lemma.
\begin{lemma} \label{lem:four}
Let $f = \sum_{i=1}^k a_i f_i$ where $f_i:\{-1,1\}^n \rightarrow \mathbb{R}$ and $a_i \in \mathbb{R}$ for all $i$. If for all $i$, $f_i$ is $(\epsilon_i, d_i)$-Fourier concentrated, then $f$ is $(\epsilon,d)$-Fourier concentrated for $\epsilon = \sqrt{\left(\sum_{i=1}^k a_i^2 \right)\left(\sum_{j=1}^k \epsilon_j^2\right)}$ and $d = \max_i d_i$.
\end{lemma} 

Many concept classes are known to be approximated by low-degree Fourier polynomials.  Combining Theorem \ref{thm:four} and Lemma \ref{lem:four}, we immediately obtain the following learning results in the probabilistic concept model whose running time matches or improves their best-known PAC counterparts:

\begin{itemize} 

\item $u(\mathsf{AC^0}$), generalizing majorities of constant depth circuits \cite{JKS}.

\item $u(\mathsf{LTF}$), generalizing majorities of linear threshold functions \cite{KKMS}.

\item $u(\mathsf{SM})$, generalizing submodular functions \cite{CKKL}.

\end{itemize}

As a further application, we can learn majorities of $k$ halfspaces with respect to the uniform distribution in time $n^{O(k^2/\epsilon^2)}$ for any $\epsilon > 0$ (choose $a$ with each entry $1/k$ and let $u$ smoothly interpolate majority with Lipschitz constant $k$).  This improves on the best known bound of $n^{O(k^4/\epsilon^2)}$ \cite{KKMS}\footnote{Recent work due to Kane \cite{Kane14} does not apply to majorities of halfspaces, only intersections.}.

Using the fact that the $\mathsf{AND}$ function has a {\em uniform} approximator of degree $O(\sqrt{n} \log(1/\epsilon))$ \cite{Paturi}, we immediately obtain a $2^{\tilde{O}(\sqrt{n})}$ time algorithm for {\em distribution-free} learning of $u(\mathsf{AND})$ in the probabilistic concept model (this class includes the set of all polynomial-size DNF formulas).  The problem of generalizing the $2^{\tilde{O}(n^{1/3})}$-time algorithm for distribution-free PAC learning of DNF formulas due to Klivans and Servedio \cite{KS01} remains open.

\subsection{Learning Monotone Functions of Halfspaces with a Margin as
  Probabilistic Concepts}
In this section we consider the problem of learning a smooth combining function $u$ of $k$ halfspaces with a margin $\rho$.  We assume that all examples lie on the unit ball $\S^{n-1}$ and that for each weight vector $w$, $||\textbf{w}|| = 1$.  For simplicity we also assume each halfspace is origin-centered, i.e. $\theta = 0$ (though our techniques easily handle the case of nonzero $\theta$).  

\begin{theorem}\label{thm:fhs}
Consider samples $(\textbf{x}_\textbf{i},y_i)_{i=1}^m$ drawn iid from distribution $\mathcal{D}$ on $\S^{n-1} \times [0,1]$ such that $\E[y|\textbf{x}] = u\left(\sum_{i=1}^t \textbf{a}_i h_i(\textbf{x})\right)$ where $u : \mathbb{R}^n \rightarrow [0,1]$ is a $L$-Lipschitz non-decreasing function, $h_i$ are origin-centered halfspaces with margin $\rho$ on $\mathcal{X}$ and $||\textbf{a}||_1 = A$.  There exists an algorithm that outputs a hypothesis $h$ such that with probability $1-\delta$, 
\[
\E_{\textbf{x}, y \sim \mathcal{D}} \left[\left(h(\textbf{x}) - u\left(\sum_{i=1}^t \textbf{a}_i h_i(\textbf{x})\right)\right)^2 \right]\leq \epsilon
\]
for $m= \left(\frac{LA}{\epsilon}\right)^{O(1/\rho)}\log(1/\delta)$.  The algorithm runs in time polynomial in $m$ and $n$.
\end{theorem}

\noindent \textbf{Remark.} If $\E[y|\textbf{x}] = u\left(\frac{1}{t}\sum_{i=1}^t h_i(\textbf{x})\right)$, that is, a function of the fraction of true halfspaces, then the run-time is {\em independent of the number of halfspaces} $t$. This holds since $A = 1$ in this case.

We now show that Theorem \ref{thm:fhs} immediately implies the first {\em polynomial-time} algorithm for PAC learning intersections of halfspaces with a (constant) margin. Consider $t$-halfspaces $\{h_1, \ldots, h_t\}$. An intersection of these $t$-halfspaces is given by $f_{\mathsf{AND}}(\textbf{x}) = \wedge_{i=1}^t h_{i}(\textbf{x})$. 

\begin{corollary}
There exists an algorithm that PAC learns any intersection of $t$-halfspaces with margin $\rho>0$ on $\S^{n-1}$ in time $\mathsf{poly}\left(n, \left(\frac{t}{\epsilon}\right)^{(C/\rho)}, \log(1/\delta) \right)$ for some constant $C$.
\end{corollary}

This result improves the previous best bound due to Klivans and Servedio \cite{klivans2008learning} that had (for constant $\rho$) a quasipolynomial dependence on the number of halfspaces $t$. 


Klivans and Servedio used random projection along with kernel perceptron and the complete quadratic kernel to obtain their results.  Here we directly use the multinomial kernel, which takes advantage of how the polynomial approximator's weights can be embedded into the corresponding RKHS.  We remark that if we are only interested in the special case of PAC learning an intersection of halfspaces with a margin (as opposed to learning in the probabilistic concept model), we can use kernel perceptron along with the multinomial kernel (and a Chebsyshev approximation that will result in an improved $O(1/\sqrt{\rho})$ dependence), as opposed to Alphatron in conjunction with the multinomial kernel. 



\section{Multiple Instance Learning}
Multiple Instance Learning (MIL) is a generalization of supervised classification in which a label is assigned to a \emph{bag}, that is, a set of instances, instead of an individual instance \cite{dietterich1997solving}. The bag label is induced by the labels of the instances in it.  The goal we focus on in this work is to label future bags of instances correctly, with high probability.  We refer the reader to \cite{amores2013multiple, herrera2016multiple} for an in-depth study of MIL. In this section we apply the previously developed ideas to MIL and give the first provable learning results for concrete schemes that do not rely on unproven assumptions.  \\

\noindent \textbf{Comparison to Previous Work.} Under the standard MI assumption, various results are known in the PAC learning setting.  Blum and Kalai \cite{blum1998note} showed a simple reduction from PAC learning MIL to PAC learning with one-sided noise under the assumption that the instances in each bag were drawn independently from a distribution.  Sabato and Tishby \cite{sabato2012multi} removed the independence assumption and gave sample complexity bounds for learning future bags. 
All the above results require the existence of an algorithm for PAC learning with one-sided noise, which is itself a challenging problem and not known to exist for even simple concept classes. 

In this work, we do not assume instances within each bag are independently distributed, and we do not require the existence of PAC learning algorithms for one-sided noise.  Instead, we give efficient algorithms for labeling future bags when the class labeling instances is an unknown halfspace with a margin or an unknown depth-two neural network.   We succeed with respect to general monotone, smooth combining functions. \\


\noindent\textbf{Notation.} Let us denote the space of instances as $\mathcal{X}$ and the space of bags as $\Beta \subseteq \mathcal{X}^*$. Let $N$ be an upper bound on the size of the bags, that is, $N = \max_{\beta \in \Beta} |\beta|$. Let the instance labeling function be $c: \mathcal{X} \rightarrow \mathbb{R}$ and the bag labeling function be $f_{\mathsf{bag}}$. We assume a distribution $\mathcal{D}$ over the bags and allow the instances within the bag to be dependent on each other. We consider two variants of the relationship between the instance and bag labeling functions and corresponding learning models.

\subsection{Probabilistic MIL}
We generalize the deterministic model to allow the labeling function to induce a probability distribution over the labels. This assumption seems more intuitive and less restrictive than the deterministic case as it allows for noise in the labels.
\begin{definition}[Probabilistic MI Assumption]
Given combining function $u: \mathbb{R} \rightarrow [0,1]$, for bag $\beta$, $f_{\mathsf{bag}}(\beta)$ is a random variable such that $Pr[f_{\mathsf{bag}}(\beta)=1] =  u\left(\frac{1}{|\beta|} \cdot \sum_{\textbf{x} \in \beta} c(\textbf{x})\right)$ where $c$ is the instance labeling function.
\end{definition}
\begin{definition}[Probabilistic MIL]
The concept class $\mathcal{C}$ is  $(\epsilon, \delta)$-Probabilistic MIL for $u$ with sample complexity $M$ and running time $T$ if under the probabilistic MI assumption for $u$, there exists an algorithm ${\cal A}$ such that for all $c \in \mathcal{C}$ as the instance labeling function and any distribution $\mathcal{D}$ on $\Beta$, ${\cal A}$ draws at most $M$ iid bags  and runs in time at most $T$ to return a bag-labeling hypothesis $h$ such that with probability $1 - \delta$, 
\[
\E_{\beta\sim \mathcal{D}}\left[\left(h(\beta) - u\left(\frac{1}{|\beta|} \cdot \sum_{\textbf{x} \in \beta} c(\textbf{x})\right) \right)^2\right] \leq \epsilon.
\]
\end{definition}

The following is our main theorem of learning in the Probabilistic MIL setting.
\begin{theorem}\label{thm:pmil}
The concept class $\mathcal{C}$ is $(\epsilon, \delta)$-Probabilistic MIL for monotone $L$-Lipschitz $u$ with sample complexity $(\frac{BL}{\epsilon})^2 \cdot \log (1/\delta)$ and running time $\mathsf{poly}(n, B, L, 1/\epsilon, \log(1/\delta))$ if all $c \in \mathcal{C}$ are $(\epsilon/CL, B)$-uniformly approximated by some kernel $\mathcal{K}$ for large enough constant $C>0$.
\end{theorem}

Combining Theorem \ref{thm:pmil} with learnability Theorems \ref{thm:sss} and \ref{thm:fhs} we can show the following polynomial time Probabilistic MIL results.
\begin{corollary}
For any monotone $L$-Lipschitz function $u$, the concept class of sigmoids over $\S^{n-1}$ are $(\epsilon, \delta)$-Probabilistic MIL with sample complexity and running time $\mathsf{poly}(n, L, 1/\epsilon, \log(1/\delta))$.
\end{corollary}
\begin{corollary}
For any monotone $L$-Lipschitz function $u$, the concept class of halfspaces with a constant margin over $\S^{n-1}$ are $(\epsilon, \delta)$-Probabilistic MIL with sample complexity and running time $\mathsf{poly}(n, L, 1/\epsilon, \log(1/\delta))$.
\end{corollary}



\bibliographystyle{alpha}
\bibliography{references}


\appendix

\section{Background}
\subsection{Learning Models}
We consider two learning models in our paper, the standard Probably Approximately Correct (PAC) learning model and a relaxation of the standard model, the Probabilistic Concept (p-concept) learning model. For completeness, we define the two models and refer the reader to \cite{valiant1984theory,kearns1990efficient} for a detailed explanation.
\begin{definition}[PAC Learning \cite{valiant1984theory}] 
	We say that a concept class $\mathcal{C} \subseteq
        \{0,1\}^{\mathcal{X}}$ is Probably Approximately Correct (PAC)
        learnable, if there exists an algorithm $\mathcal{A}$ such
        that for every $c \in \mathcal{C}$,$\delta, \epsilon >
	0$ and $\mathcal{D}$ over $\mathcal{X}$, if $\mathcal{A}$ is given access
	to examples drawn from $\mathcal{D}$ and labeled according to $c$, $\mathcal{A}$ outputs a hypothesis $h : \mathcal{X}
	\rightarrow \{0,1\}$, such that with probability at least $1 - \delta$,
	\begin{equation}
		  Pr_{\textbf{x} \sim \mathcal{D}} [h(\textbf{x}) \neq c(\textbf{x})] \leq \epsilon.
	\end{equation}
	Furthermore, we say that $\mathcal{C}$ is \emph{efficiently PAC learnable to error
	$\epsilon$} if $\mathcal{A}$ can output an $h$ satisfying the
      above with running time and sample complexity polynomial in $n$, $1/\epsilon$, and
	$1/\delta$.
\end{definition}

\begin{definition}[p-concept Learning \cite{kearns1990efficient}] 
	We say that a concept class $\mathcal{C} \subseteq
        \mathcal{Y}^{\mathcal{X}}$ is Probabilistic Concept
        (p-concept) learnable, if there exists an algorithm
        $\mathcal{A}$ such that for every $\delta, \epsilon >
	0$, $c \in \mathcal{C}$ and
	distribution $\mathcal{D}$ over $\mathcal{X} \times
        \mathcal{Y}$ with $\E[y|\textbf{x}] = c(\textbf{x})$ we have
        that $\mathcal{A}$, given access
	to examples drawn from $\mathcal{D}$, outputs a hypothesis $h : \mathcal{X}
	\rightarrow \mathcal{Y}$, such that with probability at least $1 - \delta$,
	\begin{equation}
		  \E_{(\textbf{x},y) \sim \mathcal{D}} [(h(\textbf{x}) - c(\textbf{x}))^2] \leq  \epsilon.
	\end{equation}
	Furthermore, we say that $\mathcal{C}$ is \emph{efficiently p-concept learnable to error
	$\epsilon$} if $\mathcal{A}$ can output an $h$ satisfying the
      above with running time and sample complexity polynomial in $n$, $1/\epsilon$, and
	$1/\delta$.
\end{definition}
Here we focus on square loss for p-concept since an efficient algorithm for square-loss implies efficient algorithms of various other standard losses.

\subsection{Generalization Bounds}
The following standard generalization bound based on Rademacher complexity is useful for our analysis. For a background on Rademacher
complexity, we refer the readers to \cite{BM:2002}.

\begin{theorem}[\cite{BM:2002}] \label{generalizationbound}
	Let $\mathcal{D}$ be a distribution over $\mathcal{X} \times \mathcal{Y}$ and let $\mathcal{L} : \mathcal{Y}^\prime
	\times \mathcal{Y}$ (where $\mathcal{Y} \subseteq \mathcal{Y}^\prime \subseteq \mathbb{R}$) be a
	$b$-bounded loss function that is $L$-Lipschitz in its first argument.  Let
	$\mathcal{F} \subseteq (\mathcal{Y}^\prime)^\mathcal{X}$ and for any $f \in \mathcal{F}$, let $\mathcal{L}(f; \mathcal{D}) := \E_{(\textbf{x}, y)
	\sim \mathcal{D}}[\mathcal{L}(f(\textbf{x}), y)]$ and $\hat{\mathcal{L}}(f; S) := \frac{1}{m} \sum_{i = 1}^m
	\mathcal{L}(f(\textbf{x}_\textbf{i}), y_i)$, where $S = ((\textbf{x}_\textbf{1}, y_1), \ldots,  (\textbf{x}_\textbf{m}, y_m))  \sim
	\mathcal{D}^m$. Then for any $\delta > 0$, with probability at least $1 - \delta$
	(over the random sample draw for $S$), simultaneously for all $f \in
	\mathcal{F}$, the following is true:
	\[
		|\mathcal{L}(f; \mathcal{D}) - \hat{\mathcal{L}}(f; S)| \leq 4 \cdot L \cdot \mathcal{R}_\textbf{m}(\mathcal{F})
		+ 2\cdot b \cdot \sqrt{\frac{\log (1/\delta)}{2m}}
	\]
	where $\mathcal{R}_\textbf{m}(\mathcal{F})$ is the Rademacher complexity of the function class $\mathcal{F}$. 
\end{theorem}

For a linear concept class, the Rademacher complexity can be bounded as follows.

\begin{theorem}[\cite{KST:2008}] \label{rademachercomplexity}
	Let $\mathcal{X}$ be a subset of a Hilbert space equipped with inner product $\langle
	\cdot, \cdot \rangle$ such that for each $\textbf{x} \in \mathcal{X}$, $\langle \textbf{x}, \textbf{x}
	\rangle \leq X^2$, and let $\mathcal{W} = \{ \textbf{x} \mapsto \langle \textbf{x} , \textbf{w} \rangle
	~|~ \langle \textbf{w}, \textbf{w} \rangle \leq W^2 \}$ be a class of linear functions.
	Then it holds that
	\[
		\mathcal{R}_\textbf{m}(\mathcal{W}) \leq X \cdot W \cdot \sqrt{\frac{1}{m}}.
	\]
\end{theorem}

The following result is useful for bounding the Rademacher complexity of a smooth function of a concept class.

\begin{theorem}[\cite{BM:2002, LT:1991}]
\label{rademachercomplexity2}
	Let $\phi : \mathbb{R} \rightarrow \mathbb{R}$ be  $L_{\phi}$-Lipschitz
	and suppose that $\phi(0) = 0$. Let $\mathcal{Y} \subseteq \mathbb{R}$, and for a function $f \in \mathcal{Y}^{\mathcal{X}}$. 
	Finally, for $\mathcal{F} \subseteq \mathcal{Y}^{\mathcal{X}}$, let $\phi \circ \mathcal{F} = \{\phi \circ f \colon f \in \mathcal{F}\}$.
	It holds that $\mathcal{R}_\textbf{m}(\phi
	\circ \mathcal{F}) \leq 2 \cdot L_{\phi} \cdot \mathcal{R}_\textbf{m}(\mathcal{F})$.
\end{theorem}

\subsection{Kernel Methods}
We assume the reader has a basic working knowledge of kernel methods
(for a good resource on kernel methods in machine learning we refer the reader to
\cite{scholkopf2001learning}).  We denote a kernel function by $\mathcal{K}(\textbf{x},\textbf{x}') = \langle \psi(\textbf{x}), \psi(\textbf{x}') \rangle$ where $\psi$ is
the associated feature map and $\mathcal{H}$ is the
corresponding reproducing kernel Hilbert space (RKHS). 

Here we define two kernels and a few of their properties that we will use for our
analysis. First, we define a variant of the polynomial kernel, the
{\em multinomial kernel} due to Goel et al. \cite{goel2016reliably}:

\begin{definition}[Multinomial Kernel \cite{goel2016reliably}] \label{kernel1}
	Define $\psi_d \colon \mathbb{R}^n \to \mathbb{R}^{N_d}$, where $N_d = 1 + n + \cdots +
	n^d$, indexed by tuples $(k_1, \ldots, k_j) \in [n]^j$ for each $j \in \{0,
	1, \ldots, d \}$, where the entry of $\psi_d(\textbf{x})$ corresponding to tuple
	$(k_1, \ldots, k_j)$ equals $\textbf{x}_{k_1} \cdots \textbf{x}_{k_j}$.  (When $j = 0$ we have
	an empty tuple and the corresponding entry is $1$.) Define kernel ${\mathcal{MK}_d}$
	as follows:
	\[
		{\mathcal{MK}_d}(\textbf{x}, \textbf{x}^\prime) = \langle \psi_d(\textbf{x}), \psi_d(\textbf{x}^\prime) \rangle = \sum_{j=0}^d (\textbf{x} \cdot \textbf{x}^\prime)^j.
	\]
Also define $\mathcal{H}_{{\mathcal{MK}_d}}$ to be the corresponding RKHS.
\end{definition}

It is easy to see that the multinomial kernel is efficiently computable. A multivariate polynomial $p$ of degree $d$ can be represented as an element $\textbf{v} \in \mathcal{H}_{{\mathcal{MK}_d}}$. Also, every $\textbf{v} \in \mathcal{H}_{{\mathcal{MK}_d}}$ can be interpreted as a multivariate polynomial of degree $d$ such that
\[
p(\textbf{x}) = \langle \textbf{v}, \psi_d(\textbf{x})\rangle = \sum_{\substack{(i_1, \dots, i_n) \in \{0, \ldots, d\}^n \\ i_i + \dots + i_n \leq d}} \beta(i_1, \dots, i_n) \textbf{x}_{1}^{i_1} \cdots \textbf{x}_n^{i_n}.
\]
where coefficient $\beta(i_1, \dots, i_n)$ is as follows,
\[
\beta(i_1, \dots, i_n) = \sum_{\substack{k_1, \ldots, k_j \in [n]^j
\\ j \in \{0, \ldots, d\}}}
\textbf{v}(k_1, \ldots, k_j).
\]
Here, $\textbf{v}(\cdot)$ is used to index the corresponding entry
in $\textbf{v}$. 

The following lemma is due to \cite{goel2016reliably}, following an
argument of Shalev-Shwartz et al. \cite{SSSS}: 

\begin{lemma}\label{lem:Bbound}
	Let $p(t) = \sum_{i=0}^d \beta_i t^i$ be a given univariate polynomial with
	$\sum_{i=1}^d \beta_i^2 \leq B^2$. For $\textbf{w}$ such that
        $||\textbf{w}|| \leq 1$, the polynomial $p(\textbf{w} \cdot
        \textbf{x})$ equals $\langle p_\textbf{w}, \psi(\textbf{x}) \rangle$ for some $p_\textbf{w} \in \mathcal{H}_{{MK_d}}$ with $|| p_{\textbf{w}} || \leq B$.
\end{lemma}
\noindent \textbf{Remark.} Observe that we can normalize the
  multinomial feature map such that $\forall \textbf{x} \in \mathcal{X}, \mathcal{MK}_d(\textbf{x},
  \textbf{x}) \leq 1$ for bounded space
  $\mathcal{X}$. More formally, $\max_{\textbf{x} \in \mathcal{X}}
  \mathcal{MK}_d(\textbf{x},\textbf{x}) = \max_{\textbf{x} \in \mathcal{X}} \sum_{j=0}^d
  ||\textbf{x}||^j \leq \sum_{j=0}^d
  X^j$ where $X =
  \max_{\textbf{x} \in \mathcal{X}} ||\textbf{x}||$, hence we can normalize using this value. Subsequently, in the above, $||p_{\textbf{w}}||$ will
  need to be multiplied by the same value. For
  $\mathcal{X} = \S^{n-1}$, the scaling factor is $d+1$
  \cite{goel2016reliably}.  Throughout the paper, we will assume the
  kernel to be normalized as discussed. \\

For our results on Multiple Instance Learning, we make use of the
following known kernel defined over sets of vectors:

\begin{definition}[Mean Map Kernel \cite{smola2007hilbert}]
Let $\mathcal{X}^*$ denote the Kleene closure of $\mathcal{X}$. 

The mean map kernel $\mathcal{K}_{\mathsf{mean}} \colon \mathcal{X}^* \times \mathcal{X}^* \rightarrow \mathbb{R}$ of kernel $\mathcal{K} \colon \mathcal{X} \times \mathcal{X} \rightarrow \mathbb{R}$ with feature vector $\psi \colon \mathcal{X} \rightarrow \mathcal{Z}$ is defined as,
\[
\mathcal{K}_{\mathsf{mean}}(S, T) = \frac{1}{|S||T|} \sum_{\textbf{s} \in S, \textbf{t} \in T} \mathcal{K}(\textbf{s},\textbf{t}).
\]
The feature map $\psi_{\mathsf{mean}} \colon \mathcal{X}^* \rightarrow \mathcal{Z}$ corresponding to this kernel is given by
\[
\psi_{\mathsf{mean}}(S) = \frac{1}{|S|} \sum_{\textbf{s} \in S} \psi(\textbf{s}).
\]
Also define $\mathcal{H}_{\mathsf{mean}}$ to be the corresponding RKHS.
\end{definition}
\noindent \textit{\textbf{Fact.} If $\forall \textbf{x}, \textbf{x}^\prime \in \mathcal{X}, \mathcal{K}(\textbf{x},\textbf{x}^\prime) \leq M$ then $\forall S, S^\prime \in \mathcal{X}^*, \mathcal{K}_{\mathsf{mean}}(S,S^\prime) \leq M$.}

\subsection{Approximation Theory}
We will make use of a variety of tools from approximation theory to
obtain specific embeddings of function classes into a RKHS.   The
following lemma for approximating the Boolean $\mathsf{sign}$ function
was given by \cite{daniely2015ptas}:

\begin{lemma}\label{lem:approxsign}
Let $a, \gamma, \tau > 0$. There exists a polynomial $p$ of degree $O\left( \frac{1}{\gamma} \cdot \log \left(\frac{1}{\tau}\right) \right)$ such that 
\begin{itemize}
\item For $x \in [−a, a]$, $|p({x})| < 1 + \tau$.
\item For ${x} \in [−a, a] \backslash [−\gamma \cdot a, \gamma \cdot a]$, $|p({x}) − \mathsf{sign}({x})| < \tau$.
\end{itemize}
\end{lemma}
The above lemma assumes $\mathsf{sign}$ takes on values $\{\pm 1 \}$,
but a simple linear transformation also works for $\{0,1\}$.

\cite{goel2016reliably} showed that activation functions sigmoid: $\sigma_{sig}(a) = \frac{1}{1 + e^{-a}}$ and ReLU: $\sigma_{relu}(a) = \max(0, a)$ can be $(\epsilon, B)$-uniformly approximated by the multinomial kernel for $B$ dependent on $\epsilon$, more formally they showed the following:
\begin{lemma}[Approximating a Single Hidden Unit] \label{lem:approxone}
We have,
\begin{enumerate}
\item \textbf{Sigmoid}: For all $\textbf{a} \in \S^{n-1}$ there exists a corresponding $\textbf{v} \in \mathcal{H}_{\mathcal{MK}_d}$ for $d = O(\log(1/\epsilon))$, such that 
\[
\forall\ \textbf{x} \in \S^{n-1}, |\sigma_{sig}(\textbf{a} \cdot \textbf{x}) - \langle \textbf{v}, \psi(\textbf{x}) \rangle| \leq \epsilon.
\]
Further, $||\textbf{v}|| \leq (1/\epsilon)^{O(1)}$. This implies that $\sigma_{sig}$ is $(\epsilon, (1/\epsilon)^{O(1)})$-uniformly approximated by $\mathcal{MK}_d$.
\item \textbf{ReLU}: For all $\textbf{a} \in \S^{n-1}$ there exists a corresponding $\textbf{v} \in \mathcal{H}_{\mathcal{MK}_d}$ for $d = O(1/\epsilon)$, such that
\[
\forall\ \textbf{x} \in \S^{n-1}, |\sigma_{relu}(\textbf{a} \cdot \textbf{x}) - \langle \textbf{v}, \psi(\textbf{x}) \rangle| \leq \epsilon.
\]
Further, $||\textbf{v}|| \leq 2^{O(1/\epsilon)}$. This implies that $\sigma_{relu}$ is $(\epsilon, 2^{O(1/\epsilon)})$-uniformly approximated by $\mathcal{MK}_d$.
\end{enumerate}
\end{lemma}

Finally we state the following lemmas that bound the sum of squares of coefficients of a univariate polynomial:
\begin{lemma}[\cite{sherstov2012making}] \label{lem:dbound}
	Let $p(t) = \sum_{i=0}^d \beta_i t^i$ be a univariate polynomial of degree
	$d$. Let $M$ be such that $\displaystyle\max_{t \in [-1, 1]} |p(t)| \leq M$.
	Then $\displaystyle\sum_{i=0}^d \beta_i^2 \leq (d + 1) \cdot (4e)^{2d} \cdot
	M^2$. 
\end{lemma}

\begin{lemma} \label{lem:rbound}[Fact 3 \cite{klivans2008learning}]
	Let $p(t) = \sum_{i=0}^d \beta_i t^i$ be a univariate polynomial of degree
	$d$ such that $|\beta_i| \leq M$ for all $i \in [d]$. For any $r \in \mathbb{Z}^+$ consider $p^r(t) = \left(\sum_{i=0}^d \beta_i t^i\right)^r = \sum_{i=0}^{dr} \eta_i t^i$ then,
	 $\displaystyle\sum_{i=0}^{dr} \eta_i^2 \leq (Md)^{2r}$. 
\end{lemma}
\begin{proof}
We have $p^r(t) = \sum_{i_1, \cdots, i_r \in [d]} \beta_{i_1} \cdots \beta_{i_r} t^{i_1 + \cdots + i_r}$. It follows that $\left(\sum_{i=0}^d \beta_i t^i\right)^r = \sum_{i=0}^{dr} \eta_i t^i$ is bounded above by
\[
\left(\sum_{i_1, \cdots, i_r \in [d]} |\beta_{i_1} \cdots \beta_{i_r}|\right)^2 \leq M^{2r} \left(\sum_{i_1, \cdots, i_r \in [d]}1\right)^2 = (Md)^{2r}.
\]
\end{proof}

\section{Omitted Proofs}
\subsection{Proof of Theorem 1}
Let $\Delta := \frac{1}{m}\sum_{i=1}^m(y_i- u(\langle \textbf{v}, \psi(\textbf{x}_\textbf{i})\rangle) + \xi(\textbf{x}_\textbf{i}))\psi(\textbf{x}_\textbf{i})$ and $\Delta^t := \frac{1}{m}\sum_{i=1}^m(y_i- u(\langle \textbf{v}^\textbf{t}, \psi(\textbf{x}_\textbf{i})\rangle))\psi(\textbf{x}_\textbf{i})$. Aslo define $\rho: = \frac{1}{m}\sum_{i=1}^m \xi(\textbf{x}_{\textbf{i}})^2$. We will use the following lemma:

\begin{lemma}\label{lem:alphatron}
At iteration $t$ in Alphatron, suppose $||\textbf{v}^\textbf{t} - \textbf{v}|| \leq B$ for $B > 1$, then if $||\Delta|| \leq \eta < 1$, then
\[
||\textbf{v}^\textbf{t} - \textbf{v}||^2 - ||\textbf{v}^\textbf{t+1} - \textbf{v}||^2 \geq \lambda\left(\left(\frac{2}{L} - \lambda\right)\hat{\varepsilon}(h^t) - 2\sqrt{\rho} - 2B\eta - \lambda \eta^2 - 2 \lambda\eta\right).
\]
\end{lemma}
\begin{proof}
Expanding the left hand side of the equation above, we have
\begin{align}
||\textbf{v}^\textbf{t} - \textbf{v}||^2 &- ||\textbf{v}^\textbf{t+1} - \textbf{v}||^2 \\
&= 2 \lambda \langle \textbf{v} - \textbf{v}^\textbf{t}, \Delta^t \rangle - \lambda^2||\Delta^t||^2 \label{eq:subs}\\
&= 2\lambda \langle \textbf{v} - \textbf{v}^\textbf{t}, \Delta^t - \Delta \rangle + 2\lambda \langle \textbf{v} - \textbf{v}^\textbf{t}, \Delta \rangle - \lambda^2||\Delta^t||^2\\
& \geq \frac{2\lambda}{m}\sum_{i=1}^m(u(\langle \textbf{v}, \psi(\textbf{x}_\textbf{i})\rangle + \xi(\textbf{x}_\textbf{i})) - u(\langle \textbf{v}^\textbf{t}, \psi(\textbf{x}_\textbf{i})\rangle))\langle \textbf{v} - \textbf{v}^\textbf{t}, \psi(\textbf{x}_\textbf{i})\rangle - 2 \lambda B ||\Delta|| - \lambda^2||\Delta^t||^2 \label{eq:wt1}\\
& = \frac{2\lambda}{m}\sum_{i=1}^m(u(\langle \textbf{v}, \psi(\textbf{x}_\textbf{i})\rangle + \xi(\textbf{x}_\textbf{i})) - u(\langle \textbf{v}^\textbf{t}, \psi(\textbf{x}_\textbf{i})\rangle))(\langle \textbf{v}, \psi(\textbf{x}_\textbf{i})\rangle + \xi(\textbf{x}_\textbf{i}) - \langle \textbf{v}^\textbf{t}, \psi(\textbf{x}_\textbf{i}) \rangle)\nonumber\\ 
&\ \ \ - \frac{2\lambda}{m}\sum_{i=1}^m(u(\langle \textbf{v}, \psi(\textbf{x}_\textbf{i})\rangle + \xi(\textbf{x}_\textbf{i})) - u(\langle \textbf{v}^\textbf{t}, \psi(\textbf{x}_\textbf{i})\rangle))\xi(\textbf{x}_\textbf{i})  - 2 \lambda B ||\Delta|| - \lambda^2||\Delta^t||^2 \\
& \geq \frac{2\lambda}{Lm}\sum_{i=1}^m(u(\langle \textbf{v}, \psi(\textbf{x}_\textbf{i})\rangle + \xi(\textbf{x}_\textbf{i})) - u(\langle \textbf{v}^\textbf{t}, \psi(\textbf{x}_\textbf{i})\rangle))^2 \nonumber\\ 
&\ \ \ - \frac{2\lambda}{m}\sum_{i=1}^m|u(\langle \textbf{v}, \psi(\textbf{x}_\textbf{i})\rangle + \xi(\textbf{x}_\textbf{i})) - u(\langle \textbf{v}^\textbf{t}, \psi(\textbf{x}_\textbf{i})\rangle)||\xi(\textbf{x}_\textbf{i})| - 2 \lambda B ||\Delta|| - \lambda^2||\Delta^t||^2 \label{eq:mono}\\
&\geq \frac{2\lambda}{L}\hat{\varepsilon}(h^t) - 2\lambda\sqrt{\rho} - 2\lambda B \eta - \lambda^2||\Delta^t||^2 \label{eq:final}
\end{align}
Here (\ref{eq:subs}) follows from substituting the expression of $\textbf{v}^\textbf{t+1}$, (\ref{eq:wt1}) follows from bounding $||\textbf{v}^\textbf{t} - \textbf{v}|| \leq B$ and, (\ref{eq:mono}) follows from $u$ being monotone and $L$-Lipschitz, that is, $(u(a) - u(b))(a - b) \geq \frac{1}{L}(u(a) - u(b))^2$. (\ref{eq:final}) follows from the definition of $\hat{\varepsilon}(h^t)$, the second term follows from the fact that $u$ is $[0,1]$ and $\frac{1}{m}\sum_{i=1}^m \sum_{i=1}^m |\xi(\textbf{x}_\textbf{i})| \leq \sqrt{\frac{1}{m}\sum_{i=1}^m \sum_{i=1}^m \xi(\textbf{x}_\textbf{i})^2} = \sqrt{\rho}$ and the third term is bounded using the assumption $||\Delta|| \leq \eta$.

We now bound $||\Delta^t||$ as follows.
\begin{align}
||\Delta^t||^2 &= \left|\left|\frac{1}{m}\sum_{i=1}^m(y_i- u(\langle \textbf{v}, \psi(\textbf{x}_\textbf{i})\rangle + \xi(\textbf{x}_\textbf{i})) + u(\langle \textbf{v}, \psi(\textbf{x}_\textbf{i})\rangle + \xi(\textbf{x}_\textbf{i})) - u(\langle \textbf{v}^\textbf{t}, \psi(\textbf{x}_\textbf{i})\rangle))\psi(\textbf{x}_\textbf{i})\right|\right|^2\\
&\leq ||\Delta||^2 + \left|\left|\frac{1}{m}\sum_{i=1}^m(u(\langle \textbf{v}, \psi(\textbf{x}_\textbf{i})\rangle + \xi(\textbf{x}_\textbf{i})) - u(\langle \textbf{v}^\textbf{t}, \psi(\textbf{x}_\textbf{i})\rangle))\psi(\textbf{x}_\textbf{i})\right|\right|^2 \nonumber\\ 
&\ \ \ +  2 ||\Delta|| \left|\left|\frac{1}{m}\sum_{i=1}^m(u(\langle \textbf{v}, \psi(\textbf{x}_\textbf{i})\rangle  + \xi(\textbf{x}_\textbf{i})) - u(\langle \textbf{v}^\textbf{t}, \psi(\textbf{x}_\textbf{i})\rangle))\psi(\textbf{x}_\textbf{i})\right|\right| \label{eq:tri}\\
&\leq \eta^2 + \hat{\varepsilon}(h^t) + 2\eta \label{eq:end}
\end{align}
Here (\ref{eq:tri}) follows by expanding the square and (\ref{eq:end}) follows by applying Jensen's inequality to show that for all $\textbf{a},\textbf{b} \in \mathbb{R}^m$ and vectors $\textbf{v}_\textbf{i}$ for $i \in \{1, \cdots, m\}$, $\left|\left|\frac{1}{m}\sum_{i=1}^m(\textbf{a}_\textbf{i} - \textbf{b}_\textbf{i})\textbf{v}_\textbf{i}\right|\right|^2 \leq \frac{1}{m}\sum_{i=1}^m(\textbf{a}_\textbf{i} - \textbf{b}_\textbf{i})^2||\textbf{v}_\textbf{i}||^2$ and subsequently using the fact that $||\psi(\textbf{x}_\textbf{i})|| \leq 1$. Combining (\ref{eq:final}) and (\ref{eq:end}) gives us the result.
\end{proof}
  By definition we have that $(y_i - u(\langle \textbf{v}, \psi(\textbf{x}_\textbf{i})\rangle + \xi(\textbf{x}_\textbf{i})))\psi(\textbf{x}_\textbf{i})$ are zero mean iid random variables with norm bounded by $1$. Using Hoeffding's inequality (and the fact that that the $\textbf{x}_\textbf{i}$'s are independent draws), with probability $1 - \delta$ we have
\[
||\Delta|| = \left|\left|\frac{1}{m} \sum_{i=1}^m(y_i - u(\langle \textbf{v}, \psi(\textbf{x}_\textbf{i})\rangle + \xi(\textbf{x}_\textbf{i})))\psi(\textbf{x}_\textbf{i})\right|\right| \leq  \frac{1}{\sqrt{m}} \left( 1 + \sqrt{2 \log(1/\delta)}\right).
\]
Similarly, we can bound $\rho$ using Hoeffding's inequality to get,
\[
\rho \leq \epsilon + O\left(M^2 \sqrt{\frac{\log (1/\delta)}{m}}\right) \implies \sqrt{\rho} \leq \sqrt{\epsilon} +  O\left(M \sqrt[4]{\frac{\log (1/\delta)}{m}}\right)
\]
Now using the previous lemma with $\lambda = 1/L$ and $\eta =  \frac{1}{\sqrt{m}}  \left( 1 + \sqrt{2 \log(1/\delta)}\right)$, we have
\[
||\textbf{v}^\textbf{t} - \textbf{v}||^2 - ||\textbf{v}^\textbf{t+1} - \textbf{v}||^2 \geq \frac{1}{L}\left(\frac{\hat{\varepsilon}(h^t)}{L} - 2\sqrt{\rho} - 2B\eta - \frac{ \eta^2 }{L}- \frac{2\eta}{L}\right).
\]
Thus, for each iteration $t$ of Alphatron, one of the following two cases needs to be satisfied,\\\\
\textbf{Case 1}: $||\textbf{v}^\textbf{t} - \textbf{v}||^2 - ||\textbf{v}^\textbf{t+1} - \textbf{v}||^2 \geq \frac{B \eta}{L}$\\
\textbf{Case 2}: $\hat{\varepsilon}(h^t) \leq 3BL \eta + 2L\sqrt{\rho} + \eta^2 + 2\eta =  O\left(L\epsilon + LM \sqrt[4]{\frac{\log(1/\delta)}{m}} + BL\sqrt{\frac{\log(1/\delta)}{m}}\right)$.

Let $t^*$ be the first iteration where Case 2 holds. We need to show that such an iteration exists. Assume the contradictory, that is, Case 2 fails for each iteration. Since $||\textbf{v}^\textbf{0} - \textbf{v}||^2 \leq B^2$, however, in at most $\frac{BL}{\eta}$ iterations Case 1 will be violated and Case 2 will have to be true. If $\frac{BL}{\eta} \leq T$ then $t^*$ exists such that
\[
\hat{\varepsilon}(h^{t^*}) \leq O\left(L\sqrt{\epsilon} + L M \sqrt[4]{\frac{\log(1/\delta)}{m}} + BL\sqrt{\frac{\log(1/\delta)}{m}}\right).
\]
We need to bound $\varepsilon(h^{t^*})$ in terms of $\hat{\varepsilon}(h^{t^*})$. Define $\mathcal{F} = \{ \textbf{x} \rightarrow u(\langle \textbf{z}, \psi(\textbf{x}) \rangle) : ||\textbf{z}|| \leq 2B\}$, and $\mathcal{Z} = \{ \textbf{x} \rightarrow f(\textbf{x}) - u(\langle \textbf{v}, \psi(\textbf{x}) \rangle + \xi(\textbf{x})) : f \in \mathcal{F} \}$. Using Theorem \ref{rademachercomplexity} and \ref{rademachercomplexity2} we have $\mathcal{R}_m(\mathcal{F}) = O(BL \sqrt{1/m})$. By definition of Rademacher complexity, we have
\begin{align*}
\mathcal{R}_m(\mathcal{Z}) &= \mathbb{E}_{\textbf{x}_\textbf{i}, \sigma_i}\left[\sup_{z \in \mathcal{Z}} \left(\frac{2}{m} \sum_{i=1}^m \sigma_i z(\textbf{x}_\textbf{i})\right)\right]\\
& = \mathbb{E}_{\textbf{x}_\textbf{i}, \sigma_i}\left[\sup_{f \in \mathcal{F}} \left(\frac{2}{m} \sum_{i=1}^m \sigma_i (f(\textbf{x}_\textbf{i}) - u(\langle \textbf{v}, \psi(\textbf{x}_\textbf{i}) \rangle + \xi(\textbf{x}_\textbf{i})))\right)\right]\\
& = \mathbb{E}_{\textbf{x}_\textbf{i}, \sigma_i}\left[\sup_{f \in \mathcal{F}} \left(\frac{2}{m} \sum_{i=1}^m \sigma_i f(\textbf{x}_\textbf{i}) \right)\right] - \mathbb{E}_{\textbf{x}_\textbf{i}}\left[\frac{2}{m} \sum_{i=1}^m  \mathbb{E}_{\sigma_i}[\sigma_i] u(\langle \textbf{v}, \psi(\textbf{x}_\textbf{i}) \rangle + \xi(\textbf{x}_\textbf{i}))\right]\\
& = \mathcal{R}_m(\mathcal{F})
\end{align*}
Here, $\sigma_i \in \{\pm 1\}$ are iid Rademacher variables hence $\forall\ i, E[\sigma_i] = 0$ and $\textbf{x}_\textbf{i}$ are drawn iid from $\mathcal{D}$.

Recall that $h^{t^*}(\textbf{x}) = u(\langle \textbf{v}^\textbf{T}, \psi(\textbf{x}) \rangle)$ is an element of $\mathcal{F}$ as $||\textbf{v}^\textbf{T} - \textbf{v}||^2 \leq B^2$ (case 1 is satisfied in iteration $t^*-1$) and $||\textbf{v}|| \leq B$. A direct application of Theorem \ref{generalizationbound} on $\mathcal{Z}$ with loss function $\mathcal{L}(a, \cdot) = a^2$, gives us, with probability $1- \delta$,
\[
\varepsilon(h^{t^*}) \leq \hat{\varepsilon}(h^{t^*}) + O\left(BL \sqrt{\frac{1}{m}}
        +  \sqrt{\frac{\log (1/\delta)}{m}}\right) =  O\left(L\sqrt{\epsilon} + L M \sqrt[4]{\frac{\log(1/\delta)}{m}} + BL\sqrt{\frac{\log(1/\delta)}{m}}\right).
\]

The last step is to show that we can indeed find a hypothesis satisfying the above guarantee. Since for all $h$, $\varepsilon(h)$ is up to constants equal to $err(h)$ we can do so by choosing the hypothesis with the minimum $err$ using a fresh sample set of size $O(\log(T/\delta)/\epsilon_0^2)\leq N$. This holds as given the sample size, by Chernoff bound using the fact that $\hat{\varepsilon}(h^t)$ is bounded in $[0,1]$, each $h^t$ for $t \leq T$ will have empirical error within $\epsilon_0$ of the true error with probability $1 - \delta/T$ and hence all will simultaneously satisfy this with probability $1 - \delta$. Setting $\epsilon_0 = 1/\sqrt{m}$ will give us the required bound.

\subsection{Proof of Theorem 4}
We first extend the approximation guarantees to linear combinations of function classes using the following lemma.
\begin{lemma}\label{lem:sum}
If for all $i \in [k]$, $f_i$ is $(\epsilon,B)$-uniformly approximated in kernel $\mathcal{K}$ then $\sum_{i=1}^k \textbf{a}_i f_i(\textbf{x})$ for $\textbf{a} \in \mathbb{R}^k$ such that $||\textbf{a}||_1 \leq W$ is $(\epsilon W, WB)$-uniformly approximated in kernel $\mathcal{K}$.
\end{lemma}
\begin{proof}
We have for each $i \in [k]$, $\forall \textbf{x} \in \mathcal{X},  |f_i(\textbf{x}) - \langle \textbf{v}_\textbf{i}, \psi(\textbf{x})\rangle| \leq \epsilon$ for some $\textbf{v}_\textbf{i} \in \mathcal{H}$ such that $||\textbf{v}_\textbf{i}|| \leq B$. Consider $\textbf{v} = \sum_{i=1}^k \textbf{a}_i \textbf{v}_\textbf{i}$. We have $\forall \textbf{x} \in \mathcal{X}$,
\[
\left|\sum_{i=1}^k a_i f_i(\textbf{x}) -  \langle \textbf{v}, \psi(\textbf{x})\rangle \right| = \left|\sum_{i=1}^k \textbf{a}_i (f_i(\textbf{x}) -  \langle \textbf{v}_\textbf{i}, \psi(\textbf{x})\rangle) \right| \leq \sum_{i=1}^k|a_i||f_i(\textbf{x}) -  \langle \textbf{v}_\textbf{i}, \psi(\textbf{x})\rangle| \leq \epsilon ||\textbf{a}||_1 = \epsilon W.
\]
Also $||\textbf{v}|| = \left|\left|\sum_{i=1}^k \textbf{a}_i \textbf{v}_\textbf{i}\right|\right| \leq \sum_{i=1}^k |\textbf{a}_i| ||\textbf{v}_\textbf{i}|| \leq WB$. Thus $\textbf{v}$ satisfies the required approximation. 
\end{proof}
Combining Lemmas \ref{lem:approxone} and \ref{lem:sum} we have that $\mathcal{N}_1$ for activation function $\sigma_{sig}$ is $(\epsilon_0\sqrt{k}, (\sqrt{k}/\epsilon_0)^C)$-uniformly approximated by some kernel $\mathcal{MK}_d$ with $d = O(\log(1/\epsilon_0))$ and sufficiently large constant $C>0$. Thus by Theorem \ref{thm:main}, we have that there exists an algorithm that outputs a hypothesis $h$ such that, with probability $1-\delta$,  
\[
\E_{\textbf{x}, y \sim \mathcal{D}} (h(\textbf{x}) - \mathcal{N}_2(\textbf{x}))^2 \leq C^\prime L\left(\epsilon_0\sqrt{k} + \left(\frac{\sqrt{k}}{\epsilon_0}\right)^{C} \cdot \sqrt{\frac{\log(1/\delta)}{m}}\right)
\]
for some constants $C^\prime > 0$. Setting $\epsilon_0= \frac{\epsilon}{2\sqrt{k}C^\prime L}$ and $m =  \left(\frac{2kCL}{\epsilon}\right)^{2C}\cdot\left(\frac{4\log(1/\delta)}{\epsilon^2}\right)$ gives us the required result (the claimed bounds on running time also follow directly from Theorem \ref{thm:main}).

\subsection{Proof of Theorem 6}
Let $\varepsilon(h) = \E_{x \in \{-1,1\}^n}[(u(P(x)) - u(h(x)))^2] = || u \circ P - u \circ h||_2^2$. Similar to Alphatron, we will show that with each iteration $t$ we will move closer to the true solution as long as $\varepsilon(P_t)$ is large.
\begin{lemma}
For a suitable choice of $\theta$, $||P_t - P||_2^2 - ||P_{t+1} - P||_2^2 \geq \frac{2 \lambda}{L}(\varepsilon(P_t) - L \lambda)$.
\end{lemma}
\begin{proof}
Let us define the following polynomials for $t \leq T$, $Q_t^\prime = P_{t-1} + \lambda (u \circ P - u \circ P_{t-1})\text{ and }
Q_t = \proj(Q_t^\prime)$. For all $t \leq T$,
\begin{align*}
P_t^\prime - Q_t^\prime &= (P_{t-1} + \lambda \mathsf{KM}(u \circ P - u \circ P_{t-1}, \theta)) - (P_{t-1} + \lambda(u \circ P - u \circ P_{t-1}))\\
&= \lambda(\mathsf{KM}(u \circ P - u \circ P_{t-1}, \theta) - (u \circ P - u \circ P_{t-1})). 
\end{align*}
From Lemma \ref{lem:km}, $L_\infty(\mathsf{KM}(u \circ P - u \circ P_{t-1}, \theta) - (u \circ P - u \circ P_{t-1})) \leq \theta$ implying $L_\infty(P_t^\prime - Q_t^\prime) \leq \lambda \theta \leq \theta$ since $\lambda \leq 1$.

Using Lemma \ref{lem:proj2}, we have $||\proj(P_t^\prime) - \proj(Q_t^\prime)||_2 \leq 2\sqrt{\theta k}$. Since $P_t = \mathsf{KM}(\proj(P_t^\prime), \theta)$ and $L_1(\proj(P_t^\prime)) \leq k$, using Lemma \ref{lem:kmsparse}, $||P_t - \proj(P_t^\prime)||_2 \leq \sqrt{2 \theta k}$. Using Triangle inequality, we get,
\[
||P_t - Q_t||_2 \leq ||P_t - \proj(P_t^\prime)||_2 + ||\proj(P_t^\prime) - \proj(Q_t^\prime)||_2 < 4 \sqrt{\theta k}.
\]
Observe that $||Q_t - P||_2 = L_2(Q_t - P) \leq L_1(Q_t - P) \leq L_1(Q_t) + L_1(P) \leq 2 k$. Combining these two observations, we have
\[
||P_t - P||_2^2 \leq (||P_t - Q_t||_2 + ||Q_t - P||_2)^2 \leq ||Q_t - P||^2 + 16 k \sqrt{\theta k} + 16 \theta k \leq ||Q_t - P||^2 + C k \sqrt{\theta k} \label{eq:2}
\]
for large enough constant $C>0$. Therefore, 
\begin{align}
||P_t - P||_2^2 - ||P_{t+1} - P||_2^2 &\geq ||P_t - P||_2^2  - ||Q_{t+1} - P||^2 - C k \sqrt{\theta k} \label{eq:tri}\\
& \geq ||P_t - P||_2^2  - ||Q_{t+1}^\prime - P||^2 - C k \sqrt{\theta k} \label{eq:proj}\\
& = ||P_t - P||_2^2 - ||P_t - P + \lambda (u \circ P - u \circ P_t)||_2^2 - C k \sqrt{\theta k} \label{eq:sub}\\
& = -2 \lambda (u \circ P - u \circ P_t) \cdot (P_t - P) - \lambda^2 ||u \circ P - u \circ P_t||_2^2 - C k \sqrt{\theta k} \label{eq:u}\\
& \geq \frac{2 \lambda}{L} \cdot \varepsilon(P_t) - \lambda^2 - C k \sqrt{\theta k} \nonumber.
\end{align}
Here, (\ref{eq:tri}) follows from the triangle inequality and (\ref{eq:2}), (\ref{eq:proj}) follows from projecting to a convex set reducing the distance to points in the convex set, (\ref{eq:u}) follows from $u$ being monotone, $L$-Lipschitz with output bounded in $[0,1]$. Setting $\theta$ such that $ C k \sqrt{\theta k} \leq \lambda^2$ gives the required result.
\end{proof}
As long as $\varepsilon(P_t) \geq 2 L\lambda$, we have $||P_t - P||_2^2 - ||P_{t+1} - P||_2^2 \geq 2 \lambda^2$. Since $||P_0 - P||_2^2 = ||P||_2^2 \leq L_1(P)^2 \leq k^2$, after $T = \frac{k^2}{2 \lambda^2}$, there must be some $r \leq T$ such that $||P_r - P||_2^2 \geq 2 \lambda^2$ does not hold, at this iteration, $\varepsilon(P_t) \leq 2 L \lambda = \epsilon$ for $\lambda = \frac{\epsilon}{2L}$. The last step of choosing the best $P_t$ would give us the required hypothesis (similar to Alphatron). Observe that each iteration of KMtron runs in time $\mathsf{poly}(n, k, L, 1/\epsilon)$ (Lemma \ref{lem:km}) and KMtron is run for $\mathsf{poly}(k, L, 1/\epsilon)$ iterations giving us the required runtime.

\subsection{Proof of Corollary 2}
Let $\{T_i\}_{i=1}^s$ be the ANDs corresponding to each term of the DNF. Let $T = \sum_{i=1}^s T_i$. By definition of $f$, if $f(x) = 1$ then $T(x) \geq 1$ and if $f(x) = 0$ then $T(x) = 0$. Observe that $L_1(T) \leq \sum_{i=1}^s L_1(T_i) \leq s$ using the well known fact that AND has $L_1$ bounded by $1$.

Consider the following $u$,
\[
u(a) = \begin{cases} 
      0  & a \leq 0 \\
      a & 0 < a < 1\\
      1 & a \geq 1
       \end{cases}
\]
Observe that $u$ is $1$-Lipschitz. It is easy to see that $f(x) = u(T(x))$ on $\{-1,1\}^n$. Hence, given query access to $f$ is the same as query access to $u \circ T$ over $\{-1,1\}^n$.

Since $L_1(T)$ is bounded, we can apply Theorem \ref{thm:kmtron} for the given $u$. We get that in time $\mathsf{poly}(n, s, 1/\epsilon)$, KMtron outputs a polynomial $P$ such that $\E_{x \in \{-1,1\}^n }[(u(T(x)) - u(P(x)))^2] \leq \epsilon$. Recall that $u \circ P$ may be real-valued as KMtron learns with square loss. Let us define $h(\textbf{x})$ to equal 1 if $u(P((\textbf{x})) \geq 1/2$ and 0 otherwise. We will show that $\mathbb{E}_{\textbf{x}}[h(\textbf{x}) \neq f(\textbf{x})] \leq O(\epsilon)$. Using Markov's inequality, we have $Pr_{\textbf{x}}[(u(T(\textbf{x})) - u(P(\textbf{x})))^2 \geq 1/4] \leq 4\epsilon$. For $\textbf{x}$, suppose $(u(T(\textbf{x})) - u(P(\textbf{x})))^2 < 1/4 \implies |u(T(\textbf{x}) - u(P(\textbf{x}))| < 1/2$. Since $u(T(\textbf{x})) = f(\textbf{x}) \in \{0,1\}$ then clearly $h(\textbf{x}) = f(\textbf{x})$. Thus, $Pr_{\textbf{x}}[h(\textbf{x}) \neq f(\textbf{x})] \leq 4 \epsilon$. Scaling $\epsilon$ appropriately, we obtain the required result.

\subsection{Proof of Theorem 7}
We use Lemma \ref{lem:approxsign} to show the existence of polynomial $p$ of degree $d = O\left( \frac{1}{\rho} \cdot \log \left(\frac{1}{\epsilon_0}\right) \right)$ such that for $a \in [−1, 1]$, $|p(a)| < 1 + \epsilon_0$ and for $a \in [−1, 1] \backslash [-\rho, \rho]$, $|p(a) − \mathsf{sign}(a)| < \epsilon_0$.

Since for each $i$, $\rho \leq |\textbf{w}_\textbf{i} \cdot \textbf{x}| \leq 1$, we have $|p(\textbf{w}_\textbf{i} \cdot \textbf{x}) − \mathsf{sign}(\textbf{w}_\textbf{i} \cdot \textbf{x})| \leq \epsilon_0 $ such that $p$ is bounded in $[-1,1]$ by $1 + \epsilon_0$. From Lemma \ref{lem:dbound} and \ref{lem:Bbound}, we have that for each $i$, $p(\textbf{w}_\textbf{i} \cdot \textbf{x}) = \langle \textbf{v}_\textbf{i}, \psi_d(\textbf{x}) \rangle$ such that $|| \textbf{v}_\textbf{i} || = \left(\frac{1}{\epsilon_0}\right)^{O(1/\rho)}$ where $\psi_d$ is the feature vector corresponding to the multinomial kernel of degree $d$. Using Lemma \ref{lem:sum}, we have that $\sum_{i=1}^t \textbf{a}_i h_i(\textbf{x})$ is $\left(\epsilon_0 A, A\left(\frac{1}{\epsilon_0}\right)^{O(1/\rho)}\right)$-uniformly approximated by $\mathcal{MK}_d$.


Subsequently, applying Theorem \ref{thm:main}, we get that there exists an algorithm that outputs a hypothesis $h$ such that with probability $1-\delta$,
\[
\varepsilon(h) \leq CLA\left(\epsilon_0 +  \left(\frac{1}{\epsilon_0}\right)^{C^\prime/\rho} \cdot \sqrt{\frac{\log(1/\delta)}{m}}\right)
\]
for some constants $C,C^\prime > 0$. Setting $\epsilon_0=\frac{\epsilon}{2CLA}$ and $m =  \left(\frac{2CLA}{\epsilon}\right)^{2C^\prime/\rho}\cdot\left(\frac{4\log(1/\delta)}{\epsilon^2}\right)$ to gives us the required result.

\subsection{Proof of Corollary 3}
Let $T(\textbf{x}) = \frac{1}{t}\sum_{i=1}^t \textbf{a}_i h_i(\textbf{x})$. Consider the following $u$,
\[
u(a) = \begin{cases} 
      0  & a \leq 1 - \frac{1}{t} \\
      a & 1 - \frac{1}{t} < a < 1\\
      1 & a \geq 1
       \end{cases}
\]
Observe that $u$ is $1/t$-Lipschitz. It is easy to see that $f_{\mathsf{AND}}(\textbf{x}) = u(T(\textbf{x})$. Using Theorem \ref{thm:fhs} for $L=1$ and $A = 1$, we know that there exists an algorithm that runs in time $\mathsf{poly}\left(n, \left(\frac{t}{\epsilon}\right)^{(C/\rho)}, \log(1/\delta) \right)$ and outputs a hypothesis $h$ such that $\E_{\textbf{x}, y \sim \mathcal{D}} \left[\left(h(\textbf{x}) - u\left(\sum_{i=1}^t \textbf{a}_i h_i(\textbf{x})\right)\right)^2 \right]\leq \epsilon$. Since $h$ is a real-valued function, we can use $\textsf{sign}(h)$ to get a 0-1 loss bound as in the proof of Corollary \ref{cor:dnf} to get the required result.
\subsection{Proof of Theorem 8}
Consider polynomial $P(\textbf{x}) = \sum_{S: |S| \leq d} \hat{f}(S) \chi_S(\textbf{x})$. We have,
\[
\E_{\mathcal{D}}[(f(\textbf{x}) - P(\textbf{x}))^2] = \E_{\mathcal{D}}\left[\left(\sum_{S:|S| > d} \hat{f}(S)\chi_S(\textbf{x})\right)^2\right] = \sum_{S:|S| > d} \hat{f}(S)^2 \leq \epsilon^2.
\]
We also know that $\hat{f}(S) \leq M$ (since $|f(\textbf{x})| \leq M$) for all $S \subseteq [n]$, thus $|f(\textbf{x}) - P(\textbf{x})| \leq |f(\textbf{x})| + |P(\textbf{x})| = O(n^dM)$ for all $\textbf{x} \in \{-1,1\}^n$. Also observe that
\[
\sum_{S: |S| \leq d} \hat{f}(S)^2 \leq \sum_{S} \hat{f}(S)^2 = \frac{1}{2^n}\sum_{\textbf{x} \in \{-1,1\}^n}f(\textbf{x}) \leq M
\]
where the equality follows from Parseval's Theorem. This implies that $f$ is $(\epsilon, \sqrt{M}, n^dM)$-approximated by kernel $\mathcal{K}$ and RKHS $\mathcal{H}$ with feature map $\psi$ of all monomials of degree $\leq d$. This kernel takes $O(n^d)$ time to compute. Now, we apply Theorem \ref{thm:main1} after renormalizing the kernel to get the required result.

\subsection{Proof of Lemma 5}
For all $S \subseteq [n]$, we have $\hat{f}(S) = \sum_{i=1}^k a_i \hat{f}_i(S)$. Let $\epsilon$ and $d$ be as in the lemma, we have,
\begin{align*}
\sum_{S:|S| > d} \hat{f}(S)^2 &= \sum_{S:|S| > d} \left(\sum_{i=1}^k a_i \hat{f}_i(S)\right)^2 \\
&\leq \sum_{S:|S| > d}\left(\sum_{i=1}^k a_i^2\right)\left(\sum_{j=1}^k\hat{f}_j(S)^2\right) \\
&\leq  \left(\sum_{i=1}^k a_i^2\right) \left(\sum_{i=1}^k\sum_{S:|S| > d_j} \hat{f}_j(S)^2\right) \\
&= \left(\sum_{i=1}^k a_i^2\right) \left(\sum_{i=1}^k \epsilon_i^2\right) = \epsilon^2.
\end{align*}
Here the first inequality follows from Cauchy-Schwarz, the second follows from rearranging the sum and using the fact that $\{S:|S| > d\} \subseteq \{S:|S| > d_j\}$ and the third follows from the Fourier concentration of each $f_i$.

\subsection{Proof of Theorem 9}
The following lemma is useful for our analysis.
\begin{lemma}\label{lem:milapprox}
Let $c$ be the instance labeling function mapping $\mathcal{X}$ to $\mathbb{R}$ such that  $c$ is $(\epsilon, B)$-uniformly approximated by kernel $\mathcal{K}$ and feature vector $\psi$. Then the function $f: \Beta \rightarrow \mathbb{R}$ given by $f(\beta) = \frac{1}{|\beta|} \cdot \sum_{\textbf{x} \in \beta} c(\textbf{x})$ is $(\epsilon, B)$-uniformly approximated by the mean map kernel of $\mathcal{K}$.
\end{lemma}
\begin{proof}
We have that $\forall \textbf{x} \in \mathcal{X}, |c(\textbf{x}) - \langle v, \psi(\textbf{x}) \rangle| \leq \epsilon$ for $v$ such that $||\textbf{v}|| \leq B$. Let $\mathcal{K}_{\textsf{mean}}$ be the mean map kernel of $\mathcal{K}$ and $\psi_{\textsf{mean}}$ be the corresponding vector. We will show that $\textbf{v}$ $(\epsilon, B)$-approximates $f$ in $\mathcal{K}_{\textsf{mean}}$. This follows from the following,
\begin{align*}
\left| f(\beta) - \langle \textbf{v}, \psi_{\textsf{mean}}(\beta) \rangle \right| = &\left| \frac{1}{|\beta|} \cdot \sum_{\textbf{x} \in \beta} c(\textbf{x}) - \frac{1}{|\beta|} \cdot \sum_{\textbf{x} \in \beta} \langle \textbf{v}, \psi(\textbf{x}) \rangle \right| \\
& \leq \frac{1}{|\beta|} \cdot \sum_{\textbf{x} \in \beta} |c(\textbf{x}) - \langle \textbf{v}, \psi(\textbf{x}) \rangle| \leq \epsilon.
\end{align*}
\end{proof}
Consider $c \in \mathcal{C}$ that is $(\epsilon/CL, B)$-uniformly approximated by some kernel $\mathcal{K}$ for large enough constant $C>0$ (to be chosen later). Using Lemma \ref{lem:milapprox} we know that $f(\beta) = \frac{1}{|\beta|} \cdot \sum_{\textbf{x} \in \beta} c(\textbf{x})$ is $(\epsilon/CL, B)$-uniformly approximated by the mean map kernel $\mathcal{K}_{\textsf{mean}}$ of $\mathcal{K}$. Applying Theorem \ref{thm:main}, we get that with probability $1 - \delta$, 
\[
\E_{\beta\sim \mathcal{D}}\left[\left(h(\beta) - u\left(\frac{1}{|\beta|} \cdot \sum_{\textbf{x} \in \beta} c(\textbf{x})\right) \right)^2\right] \leq \frac{C^\prime}{C}\epsilon.
\]
for sufficiently large constant $C^\prime>0$. Choosing $C \geq C^\prime$ gives the result.

\end{document}